\documentclass[sigplan,nonacm,10pt]{acmart}

\usepackage{tikz}
\usetikzlibrary{cd}
\tikzset{
  lang/.style={draw, rounded corners=2pt, inner sep=2.5pt,
               font=\scriptsize\ttfamily},
  hub/.style={lang, double, fill=black!8},
  tool/.style={draw, inner sep=2.5pt, font=\scriptsize\itshape},
  tcb/.style={tool, fill=black!8},
  art/.style={font=\scriptsize},
  chk/.style={->, >=stealth},
  pred/.style={->, >=stealth, dashed},
  repr/.style={->, >=stealth, dotted, thick},
  elbl/.style={font=\tiny\ttfamily, inner sep=1.5pt},
}
\usepackage{array}
\usepackage{booktabs}
\usepackage{xspace}
\usepackage{listings}

\newif\ifarxiv
\arxivfalse
\newcommand{\hg}{\textsc{hurdy-gurdy}\xspace}
\newcommand{\Prog}[1]{\mathsf{Prog}_{#1}}
\newcommand{\Beh}[1]{\mathsf{Beh}_{#1}}
\newcommand{\Obs}[1]{\mathsf{Obs}_{#1}}
\newcommand{\sem}[2]{\llbracket #1 \rrbracket_{#2}}

\newcommand{\carry}{\Lambda}          
\newcommand{\proj}{\pi}
\newcommand{\eqproj}[1]{\equiv_{#1}}
\newcommand{\dom}{\mathsf{dom}}
\newcommand{\cov}{\mathsf{cov}}
\newcommand{\keep}{\mathsf{keep}}
\newcommand{\lost}{\mathsf{lost}}
\newcommand{\tcb}{\ensuremath{\mathsf{TCB}}\xspace}
\newcommand{\tpred}{\ensuremath{\mathtt{predicted}}\xspace}
\newcommand{\tprov}{\ensuremath{\mathtt{proved}}\xspace}
\newcommand{\tchk}{\ensuremath{\mathtt{checked}}\xspace}
\newcommand{\trepr}{\ensuremath{\mathtt{reproducible}}\xspace}
\newcommand{\ttrust}{\ensuremath{\mathtt{trusted}}\xspace}
\newcommand{\unsupported}{\ensuremath{\mathtt{unsupported}}\xspace}

\arxivtrue
\newcommand{\repofootnote}{\footnote{Code, evidence, and the Lean
  mechanization: \url{https://github.com/cksystemsgroup/hurdy-gurdy}
  --- this version is tag \texttt{arxiv.2}.}}

\theoremstyle{definition}
\newtheorem{assumption}{Assumption}
\newtheorem*{genexample}{Example}
\usepackage[capitalise]{cleveref}
\crefname{assumption}{assumption}{assumptions}
\Crefname{assumption}{Assumption}{Assumptions}

\begin{document}

\title{Untrusted Authors, Trusted Answers}
\subtitle{A Calculus of Fidelity-Graded Translations}

\author{Christoph Kirsch}
\affiliation{%
  \institution{University of Salzburg, Austria}\country{}}
\affiliation{%
  \institution{Czech Technical University, Prague, Czechia}\country{}}
\email{ck@cs.uni-salzburg.at}

\begin{abstract}

To answer a question about a program, move the program to where the
question is decidable. Every such move is a translation, and every
translation is a place to be wrong. We study translation as a
\emph{graph} --- many languages, a few reasoning targets, independently
built routes of honestly different trustworthiness --- and give it a
calculus: pairs of languages close commuting squares that are
\emph{directional} (exactness is the identity-embedding special case of
over-approximation), checkable per program, and composable, a route's
contract being the componentwise meet of its hops' contracts ---
assurance class,
direction, kept observables, measured cost. One asymmetry organizes
trust: witness-carrying answers are self-certifying by replay at the
source; universal answers are where grades, independent branches, and
re-checked certificates earn their cost. The compositional core, lax
telescope included, is mechanized in Lean~4.

hurdy-gurdy implements the calculus as \emph{two planes} meeting in one
registry. The \emph{use} plane reads declarations and produces
evidence-carrying answers; its builders and its intended player are
both LLMs, untrusted by construction. The \emph{evolution} plane grows
the graph: unmet questions are recorded as demand, pairs are
\emph{recommended by evidence} and registered by humans, and a ratchet
keeps every prior verdict standing. \emph{Answers never write; growth
never answers.} Run indefinitely, the loop converges on every reducibly
decidable question, at fidelity that only rises. We measure the July
2026 snapshot --- per-construct conjoined coverage, dual-route branch
agreement for two ISAs, source-level witness replay, certified
unreachability re-validated by a formally verified checker, escape
rates for the gate itself --- and report the defects the architecture
caught in its own authors' work.

\end{abstract}

\maketitle


\section{Introduction}
\label{sec:intro}

To answer a question about a program, move the program to where the
question is decidable. A C program's reachability question becomes a
bit-vector model-checking problem; a Python function's assertion becomes
a linear-arithmetic query; a chemical reaction network's coverability
question becomes an SMT unrolling. Every such move is a translation, and
every translation is a place to be wrong. The classical responses are to
\emph{prove} the translator once and for all
\citep{leroy2009compcert,kumar2014cakeml}, or to \emph{validate} each of
its runs \citep{pnueli1998tv,necula2000tv,lopes2021alive2}. Both are
statements about a single edge.\repofootnote

This paper is about the \emph{graph}: several source languages, a few
solver-facing targets, more than one way to get from here to there ---
and, crucially, edges of \emph{honestly different} trustworthiness. A
rule-for-rule logic translator is foreseeable from its specification; an
optimizing C compiler is only pinned; two translators derived
independently from a prose manual and a formal model are neither proved
--- but their independence is itself a resource. A theory of the graph
must let each edge say what it actually guarantees, compose those
statements honestly, and let weak edges corroborate one another into
strong routes. Its stance is that of distributed systems: a translator
is a component that fails, and failure is \emph{contained} --- by
declared loss, decidable checks, and independent routes --- not
prevented by universal proof.

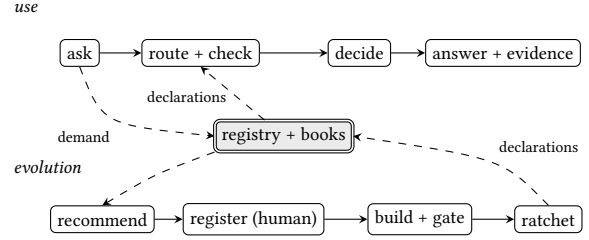
\begin{figure}
\centering
\begin{tikzpicture}[x=1cm, y=1cm,
  stg/.style={draw, rounded corners=2pt, inner sep=2.6pt, font=\scriptsize},
  flow/.style={->, >=stealth},
  data/.style={->, >=stealth, dashed},
  planelbl/.style={font=\scriptsize\itshape, anchor=west}]
\node[planelbl] at (-0.4, 1.0) {use};
\node[stg] (q)  at (0.6, 0.4)  {ask};
\node[stg] (r)  at (2.2, 0.4)  {route + check};
\node[stg] (d)  at (4.3, 0.4)  {decide};
\node[stg] (a)  at (6.2, 0.4)  {answer + evidence};
\draw[flow] (q) -- (r); \draw[flow] (r) -- (d); \draw[flow] (d) -- (a);
\node[stg, double, fill=black!8] (reg) at (3.3, -0.7) {registry + books};
\draw[data] (reg) -- (r.south) node[pos=0.45, left, font=\tiny] {declarations};
\draw[data] (q.south) .. controls (0.9, -0.55) .. (reg.west)
  node[pos=0.55, below left, font=\tiny] {demand};
\node[planelbl] at (-0.4, -1.1) {evolution};
\node[stg] (rec) at (0.9, -1.8)  {recommend};
\node[stg] (hu)  at (2.9, -1.8)  {register (human)};
\node[stg] (bg)  at (5.1, -1.8)  {build + gate};
\node[stg] (ra)  at (6.8, -1.8)  {ratchet};
\draw[flow] (rec) -- (hu); \draw[flow] (hu) -- (bg); \draw[flow] (bg) -- (ra);
\draw[data] (reg.south west) .. controls (1.6, -1.2) .. (rec.north);
\draw[data] (ra.north) .. controls (6.2, -1.0) .. (reg.east)
  node[pos=0.5, above right, font=\tiny] {declarations};
\end{tikzpicture}
\caption{The two planes, meeting in one data structure. The use plane
reads the registry's declarations and produces evidence-carrying
answers; the evolution plane turns recorded demand into new
declarations, gated and ratcheted, with registration a human act.
\emph{Answers never write; growth never answers.}}
\label{fig:planes}
\end{figure}

\paragraph{Two directions, two planes}
This work is as much an experiment about \emph{who can build and use}
such a system as about the system itself, and its two directions are
the platform's two planes (\Cref{fig:planes}). On the \emph{use} plane,
LLMs are \emph{supported}: the intended player is an LLM that chooses
questions and routes but can take no unchecked step --- every move is
deterministic, graded, cross-checked, and, for witness-carrying
answers, self-certifying at the source (a first controlled experiment:
\S\ref{sec:eval-player}). On the \emph{evolution} plane, LLMs are
\emph{utilized}: every pair was implemented by an independent LLM agent
from a one-page brief, largely unsupervised, with the architecture's
cross-checks as the only semantic gate (\S\ref{sec:bugs} reports what
the gate caught) --- and the plane is instrumented end to end:
unanswerable questions are diagnosed and recorded as demand, pairs are
\emph{recommended by evidence} before a human registers them, and the
ratchet (\Cref{prop:ratchet}) keeps every prior verdict standing
(\S\ref{sec:books}). Run indefinitely, this loop is a semi-algorithm
for answerability: in the limit it reaches every \emph{reducibly
decidable} question, at fidelity that only rises --- and it is not an
oracle: undecidability, the adequacy floor, and the finite supply of
independent semantic anchors stand at every graph size
(\Cref{sec:overview,sec:conclusion}).

\paragraph{Provenance disclaimer}
All of \hg\ --- framework, interpreters, pairs, solver plumbing,
evaluation harness, Lean mechanization --- is LLM-generated code
(Claude Opus and Claude Fable agents), and most of this paper's text is
LLM-generated too; no human reviewed either for semantic correctness.
The human contribution is the architecture: the calculus, the
contracts, and the decisions about what to build and claim. Every
quantitative claim in \Cref{sec:evaluation} regenerates from the
artifact by one script, and the mechanization's axiom audit re-runs at
every build --- the reader need not trust the writers, human or
otherwise, beyond what the thesis itself requires.

\paragraph{This paper}
The unit is the \emph{pair} (\Cref{sec:calculus}): two languages, a
pure translator, shared pure interpreters, and a pure
\emph{target-to-source interpreter} carrying target behaviors back to
source vocabulary. These close a commuting square that is
\emph{directional} --- exactness is the identity-embedding special case
of over-approximation --- and \emph{decidable per program}: run both
sides, compare under the declared projection, and a failure names its
step and observable. Squares paste (under a one-line congruence
condition on what carry-backs may read), and composition is told once:
a pair's \emph{contract} --- assurance class, direction, kept
observables, measured cost --- composes along a route by the
componentwise meet, the weakest hop on every axis at once
(\Cref{sec:fidelity}). Two mechanisms move trust where the meet says it
cannot go --- per-run re-establishment by inline checking, and
corroboration by independently derived branches whose declared
semantic artifacts are disjoint --- and one asymmetry caps the story:
existential answers are self-certifying by replay
(\Cref{thm:existential}), while universal answers are exactly where
grades, branches, and certificates earn their cost
(\Cref{thm:universal}). \Cref{sec:instantiation} describes the
platform --- 13 languages, 15 registered pairs (13 in the evaluated
snapshot), two reasoning hubs --- \Cref{sec:economy} its economy and
books, and \Cref{sec:evaluation} measures both.

\paragraph{Contributions}
\begin{itemize}
\item A calculus of directional, fidelity-graded translation pairs:
  per-program decidable squares with localization, pasting under an
  explicit support condition, exactness as the identity-embedding case
  (\Cref{prop:exactid}), and universal transfer across
  over-approximation (\Cref{prop:lax}).
\item The contract algebra: composition as one componentwise meet
  across assurance, direction, loss, and cost --- its logical
  coordinates (class and direction) mechanized in Lean~4 with the
  rest of the compositional core, lax telescope included
  (\S\ref{sec:mech}); loss composes by declared keep-sets, cost by
  measurement.
\item The existential/universal asymmetry as an architectural
  principle, with its end-to-end theorems (\S\ref{sec:endtoend}).
\item The two-plane architecture and its economy: demand books,
  evidence-backed recommendation, human-valved registration, and the
  ratchet --- the untrusted-generator stance extended from building
  pairs to choosing them (\S\ref{sec:books}).
\item A measured, LLM-built instantiation: conjoined coverage, dual
  ISA branches with disjoint decision stacks, a compliance-derived
  benchmark, witness replay, certified unreachability at up-to-megabyte
  scale under a formally verified checker, gate escape rates, a first
  LLM-player experiment, and the defects caught
  (\Cref{sec:instantiation,sec:evaluation}).
\end{itemize}

This document is a dated capability snapshot (July 2026) of a system
whose coverage ratchets monotonically by construction; every number in
\Cref{sec:evaluation} regenerates from the artifact, and the appendix
--- detailed proofs and the construct-inventory catalog --- follows the
bibliography.


\section{Overview: one question, one growth step}
\label{sec:overview}

Consider a small C program: a loop over an array with an index
computation, and a question --- can \texttt{x == 42} hold at the failure
label? \Cref{fig:run} draws the route this question takes; we narrate
it, then one growth step, and \Cref{sec:calculus} makes the
definitions precise.

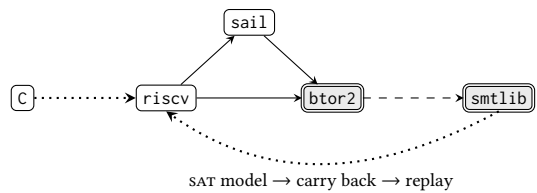
\begin{figure}
\centering
\begin{tikzpicture}[x=1cm, y=1cm]
\node[lang] (c)    at (0.0, 0.0) {C};
\node[lang] (rv)   at (1.9, 0.0) {riscv};
\node[lang] (sail) at (3.0, 1.0) {sail};
\node[hub]  (bt)   at (4.1, 0.0) {btor2};
\node[hub]  (smt)  at (6.3, 0.0) {smtlib};
\draw[repr] (c) -- (rv);
\draw[chk] (rv) -- (bt);
\draw[chk] (rv) -- (sail);
\draw[chk] (sail) -- (bt);
\draw[pred] (bt) -- (smt);
\draw[->, >=stealth, dotted, thick]
  (smt.south) .. controls (4.6, -1.1) and (3.2, -1.1) .. (rv.south)
  node[midway, below, font=\scriptsize]
  {\textsc{sat} model $\to$ carry back $\to$ replay};
\end{tikzpicture}
\caption{One run down the spine. The question travels
C$\to$RISC-V$\to$BTOR2$\to$SMT-LIB (RISC-V reaching BTOR2 twice, the
manual-derived and Sail-derived branches whose agreement corroborates
both); a \textsc{sat} model is carried back hop by hop and
\emph{replayed} in the source-side interpreter --- the answer's
evidence, not the solver's say-so (\Cref{thm:existential}).}
\label{fig:run}
\end{figure}

\paragraph{The spine}
The program first moves from C to RISC-V. The translator on this edge is
a \emph{pinned} C compiler: a digest-locked gcc with recorded flags. We
can replay it bit-for-bit, but nobody can foresee or prove its output ---
its honest grade is $\trepr$, determinism without meaning. The RISC-V
program then moves to BTOR2, a bit-vector transition-system logic that
model checkers consume, and from there across a rule-for-rule bridge to
SMT-LIB, where bit-vector solvers decide the reachability query.

Each of these moves is a \emph{pair}: translator, source interpreter,
target interpreter, and a target-to-source interpreter that carries
behaviors backwards. When the RISC-V-to-BTOR2 translator runs, the
platform can immediately check it \emph{on this very program}: interpret
the RISC-V program directly (the left edge of a square), interpret the
BTOR2 translation and carry the resulting trace back to RISC-V
observables (the other three edges), and compare, step by step, on the
declared observables --- program counter, registers, halt flag. That
check passing is what the edge's $\tchk$ grade \emph{means}. Note
carefully what it covers: every hop \emph{from the compiled binary
onward}. The C hop itself has no square --- there is no C interpreter
and no ISA-to-C carry-back --- so the verified portion of the route
begins at the binary; the pinned compiler contributes reproducibility,
re-checked separately by a corpus-level differential against an
independent C verifier (\Cref{sec:instantiation}).

\paragraph{The branch}
RISC-V reaches BTOR2 twice. One translator was written against the prose
ISA specification; a second route factors through Sail, deriving its
semantics from the official RISC-V Sail model. On the segment where they
diverge --- from RISC-V to BTOR2 --- the two routes share no translator
and no carry-back; below the reconvergence they share the hub bridge,
so it is the diverse prefix that agreement corroborates
(\S\ref{sec:branching}). The
platform runs the question down both, and the answers must agree. If they
do, two independently derived encodings of RISC-V semantics corroborate
each other --- evidence strictly stronger than either route's own grade
(\Cref{lem:agree}). If they disagree, at least one route has a defect,
and the per-hop squares localize it: which hop, which step, which
observable (\Cref{lem:disagree,cor:localization}). The same dual-route
structure exists for AArch64 (Arm manual vs.\ Arm Sail model).

\paragraph{The answer comes home}
Suppose the SMT solver answers \textsc{sat}: the label is reachable, with
a model. The model is not the answer the user asked for --- it is a
valuation of SMT variables. The route's carry-backs now run in reverse:
the SMT model becomes a BTOR2 witness, the witness a RISC-V trace, the
trace an initial memory and register state --- concretely, the inputs of
the original C program. The platform then \emph{replays} those inputs in
the route's source-side interpreter --- for this spine, the shared
RISC-V interpreter on the compiled binary, the question's observable
read at the ISA level (\S\ref{sec:composition}) --- and watches the
failure state actually happen.
At that moment the user's trust obligations collapse: even if every
translator, both routes, and the solver were wrong, the replayed run
exhibits the behavior (\Cref{thm:existential}). Witness-carrying answers
are self-certifying; only the interpreter that replays them --- itself
differentially tested against an independently generated ISA simulator ---
remains to be believed.

Had the solver answered \textsc{unsat}, no replay could vouch for it.
That is where the grades, the branch agreement, the multi-engine
corroboration, and the independently checked certificates carry the
weight (\Cref{thm:universal}) --- and where the platform spends its
assurance budget.

\paragraph{Beyond one spine}
Nothing above is specific to C or RISC-V
(\Cref{fig:registry} draws the full registry graph). The same square contract admits
eBPF, Wasm, and EVM programs into the BTOR2 hub; Python functions and
chemical reaction networks reach the SMT-LIB hub directly; a SMILES
molecular graph reaches a molecular-formula target with no solver at all
(a \emph{compile pair} --- the calculus does not care that the languages
are far from computation, only that both have defined meaning). Each
pair declares what it preserves (its projection), what it rejects (typed,
enumerable $\unsupported$), and its grade; the registry composes routes
and reports each route's composed coverage, grade, and loss.

\paragraph{One growth step}
Now ask the platform something it cannot do: is this RISC-V program's
loop \emph{live} --- does it always eventually make progress? The
diagnosis (\texttt{gurdy why-not}) walks its obstacles in order
and stops at the third: routes exist, the observables survive, but no
registered reasoning language expresses a liveness shape. The failure
is recorded in the platform's books --- the question verbatim, the
obstacle, the generation target it names (a liveness-to-safety
reduction on the bit-level hub, or a new temporal hub), an origin tag
--- and the recommendation board aggregates it beside every other
unmet demand. A generator drafts the brief citing those records; a
human reads the evidence and registers, or declines; a builder agent
implements against the contract; the gate admits or rejects; the
ratchet keeps every prior verdict standing; the question is re-asked.
That is one turn of the evolution plane, and no step of it answered a
question or was answered by one --- \emph{answers never write; growth
never answers}. The rest of
the paper makes each of these words precise, then measures the system.

\paragraph{The graph in the limit}
Nothing above says the graph is finished --- the architecture is built
so that it never has to be. When a question is unanswerable, exactly
one of four obstacles fails first, each named by one shipped call
(\texttt{gurdy why-not}): no route reaches a reasoning language; a
route exists but some hop's declared loss discards the observable the
question reads; no registered reasoning language expresses the
question's shape; or the solver exhausts its budget. The diagnosis
names the missing pair; the demand is recorded in the platform's
books and aggregated into recommendations (\S\ref{sec:books}); and
that makes a \emph{loop} well-defined: generate the recommended pair,
gate it, ratchet, re-ask. We leave the closed-loop composition as a
concept --- its stages are shipped, tested calls, and no autonomous
run is claimed or evaluated here --- but its limit
clarifies what the platform is \emph{for}. The answerable set is the
closure of a source language's questions under the registered
reductions; it only grows (\Cref{prop:ratchet}); and any question that
some finite chain of sound, deterministic reductions can bring to a
mechanized decision procedure is eventually answered, with quantified
trust. For a fixed finite-state program --- the compiled binary
above --- every question over kept observables is decidable in
principle, so growth for a \emph{given} program is a matter of
question shapes (new reasoning hubs) and of \emph{cost} (abstraction
pairs, \S\ref{sec:lax}), not connectivity. Trust grows on a separate
axis --- an added route answers no new question; it corroborates
(\S\ref{sec:branching}) --- and saturates at the supply of independent
semantic anchors, which no generator enlarges. Undecidability and
interpreter adequacy (\Cref{asm:adequacy}) stand at every graph size:
the limit is not an oracle but the reflection of everything mechanized
deciding knows how to answer, behind one interface, with every
answer's trust itself measured. The artifact's \texttt{POTENTIAL.md}
develops the limit and its boundaries in full.


\section{The machine: a calculus of pairs and its guarantees}
\label{sec:calculus}

This section is the whole use plane, in the introduction's order:
languages, behaviors, and projections (\S\ref{sec:languages}); the
directional commuting square (\S\ref{sec:pairs}); composition ---
pasting, and the contract meet (\S\ref{sec:pasting},
\S\ref{sec:grades}); the direction axis at work (\S\ref{sec:lax});
trust beyond the meet --- re-establishment, branches, and the
trusted-base ledger (\S\ref{sec:composition},
\S\ref{sec:branching}); the end-to-end asymmetry
(\S\ref{sec:endtoend}); and the mechanization (\S\ref{sec:mech}).
Coverage and the ratchet are evolution-plane machinery and live with
the economy (\S\ref{sec:coverage}).

\subsection{Languages, behaviors, projections}
\label{sec:languages}

\begin{definition}[Language]
\label{def:language}
A \emph{language} $A$ is a triple $(\Prog{A}, F_A, \sem{\cdot}{A})$ where
$\Prog{A}$ is a decidable set of \emph{programs}; $F_A$ is a finite set of
named \emph{observable fields}, each $f \in F_A$ with a value domain
$V_f$; and $\sem{\cdot}{A} : \Prog{A} \rightharpoonup \Beh{A}$ is the
\emph{reference semantics}, a partial function into behaviors.
An \emph{observable state} is a total map
$\sigma \in \Obs{A}$, where $\Obs{A} = \prod_{f \in F_A} V_f$; a
\emph{behavior} is a finite sequence of observable states,
$\Beh{A} = \Obs{A}^{*}$, recorded \emph{post-step}: the $i$-th state is
the observable state after the $i$-th transition.
\end{definition}

The only admission criterion for a language is that $\sem{\cdot}{A}$ is
\emph{definable} --- an ISA manual, a logic's model theory, a reaction-network
semantics, and a molecular-graph notation all qualify. Nothing in this
section assumes languages are executable or even computational.

The scope, stated plainly: behaviors are \emph{finite} and the
reference semantics is a \emph{function} --- nondeterminism,
divergence, and reactive I/O are outside. Behavior \emph{inclusion} is
inside in exactly one controlled form: the square is a
\emph{directional} claim (\Cref{def:pair,def:faithful}), checked as a
projected trace equality along a declared witness embedding —
two-sided ($\mathsf{exact}$) or the simulation half
($\mathsf{over}$); nothing else is admitted. This is the right notion
for near-structure-preserving translation and the wrong one for
optimizing compilation, whose home is full simulation over retimed
structure --- accordingly the one optimizing hop in the instantiation
carries no square at all, and C's non-functional official semantics
enters through a determinized fragment, policed differentially
(\Cref{sec:instantiation}).

A program here is a \emph{closed} configuration: free inputs, when present,
are part of $p$ (a program together with an input valuation). Open programs
--- programs quantified over their inputs --- enter the square through
their \emph{closing valuations} (\Cref{def:faithful} quantifies over
them) and take center stage in \S\ref{sec:endtoend}, where a solver
decides over all inputs and a witness closes one.

\begin{definition}[Interpreter]
\label{def:interpreter}
An \emph{interpreter} for $A$ is a computable partial function
$I_A : \Prog{A} \rightharpoonup \Beh{A}$ that is \emph{pure}: identical
input yields byte-identical output, on every host (the contract;
\S\ref{sec:eval-perf} enforces it by twice-and-diff on the evaluation
host). Outside its domain
it fails with a typed $\unsupported$ verdict naming the construct;
$\dom(I_A)$ is $A$'s \emph{covered fragment}.
\end{definition}

\begin{assumption}[Interpreter adequacy]
\label{asm:adequacy}
$I_A(p) = \sem{p}{A}$ for all $p \in \dom(I_A)$.
\end{assumption}

Adequacy is an assumption, not a theorem, and we keep it visible: the
irreducible residue in every trusted-base computation of
\S\ref{sec:tcb}, discharged \emph{empirically} by differential testing
against independently produced oracles (\Cref{sec:instantiation};
Python is the honest exception --- its pinned runtime is both semantics
and interpreter).

\begin{definition}[Projection]
\label{def:projection}
A \emph{projection} over $A$ is a subset $\proj \subseteq F_A$. It lifts
to states by restriction, $\proj(\sigma) = \sigma\!\restriction_{\proj}$,
and to behaviors pointwise,
$\proj(\sigma_1 \cdots \sigma_n) = \proj(\sigma_1) \cdots \proj(\sigma_n)$
(in particular, projected behaviors of different lengths are unequal).
Write $b \eqproj{\proj} b'$ for $\proj(b) = \proj(b')$.
\end{definition}

\subsection{Pairs and the directional square}
\label{sec:pairs}

\begin{definition}[Pair]
\label{def:pair}
A \emph{pair} $P : A \to B$ is a tuple
$(A, B, T, \carry, \proj_P, d_P)$ where
$T : \Prog{A} \rightharpoonup \Prog{B}$ is the \emph{translator};
$\carry : \Beh{B} \rightharpoonup \Beh{A}$ is the
\emph{target-to-source interpreter}, which re-expresses a target behavior
as a source behavior; $\proj_P \subseteq F_A$ is the pair's
\emph{declared projection} --- exactly the source observables it promises
to preserve; and $d_P \in \{\mathsf{exact}, \mathsf{over}\}$ is its
\emph{declared direction}. An $\mathsf{over}$ (over-approximating) pair
additionally ships a pure \emph{witness embedding} $W$, mapping each
closing valuation $x$ of a source program $p$ to a closing valuation
$W(p,x)$ of $T(p)$; for an $\mathsf{exact}$ pair, $W$ is the identity.
$T$, $\carry$, and $W$ are pure (as in \Cref{def:interpreter});
$I_A$ and $I_B$ are owned by the languages and shared by every pair that
touches them. The pair's domain $\dom(P)$ is the set of $p$ with
$p \in \dom(I_A) \cap \dom(T)$, $T(p) \in \dom(I_B)$, and
$I_B(T(p)) \in \dom(\carry)$. Like the projection, the direction is a
declared, \emph{protected} part of the contract.
\end{definition}

\begin{definition}[Faithfulness: the directional square]
\label{def:faithful}
\label{def:lax}
$P$ is \emph{faithful at} $p \in \dom(P)$ when its square commutes at
$p$ \emph{along the witness embedding}: for every closing valuation $x$
of $p$,
\[
\proj_P\bigl(I_A(p(x))\bigr) \;=\;
\proj_P\bigl(\carry(I_B(\,T(p)(W(p,x))\,))\bigr),
\]
where a closed program has the single trivial valuation. For an
$\mathsf{exact}$ pair ($W = \mathrm{id}$) this is the familiar
commuting square $\proj_P(I_A(p)) = \proj_P(\carry(I_B(T(p))))$:
\[
\begin{tikzcd}[column sep=large]
p \arrow[r, "T"] \arrow[d, "I_A"'] & T(p) \arrow[d, "I_B"] \\
I_A(p) & \carry(I_B(T(p))) \arrow[l, "\eqproj{\proj_P}"']
\end{tikzcd}
\]
$P$ is faithful on $C \subseteq \dom(P)$ when it is faithful at every
$p \in C$.
\end{definition}

The direction is the \emph{reading} of that equality. The embedding
selects which target valuations the square checks: exactly those in its
range. When $W$ is the identity --- surjective onto the closing
valuations of $T(p)$ --- every target valuation is checked and the
target has no behaviors beyond the source's kept ones: the square is
two-sided, a bisimulation on the kept observables, and that is all
$d_P = \mathsf{exact}$ asserts. When $W$ is not surjective, the target
valuations outside its range are the pair's \emph{added} behaviors ---
their existence is the point, not a defect --- and the same checked
equality asserts only the simulation half,
$I_A(p) \sqsubseteq_{\proj_P} \carry(I_B(T(p)))$: every source behavior
has a target counterpart on the kept observables, and the target may
have more. Exactness is thus not a separate notion but the
identity-embedding special case of one square (\Cref{prop:exactid},
machine-checked). The check itself is direction-indifferent --- run
both legs, compare under $\proj_P$ --- so everything below applies to
both directions verbatim; only what agreement \emph{means} differs,
which is why direction is protected like the projection: an
$\mathsf{over}$ square passed off as $\mathsf{exact}$ would launder
added behaviors into a meaning-preservation claim.

Three remarks. First, everything in \Cref{def:faithful} is computable:
faithfulness at a given $p$ is a \emph{decidable, executable check} --- run
both sides, compare under $\proj_P$. We call this check the
\emph{square oracle}; in fault-tolerance terms it is an acceptance test
in the sense of recovery blocks \citep{randell1975recovery}, turning a
silently wrong --- in effect Byzantine --- translator into a
\emph{fail-stop} component \citep{schlichting1983failstop}. Second, when the check fails it fails at a
\emph{first} position: a least step index $i$ and a field
$f \in \proj_P$ on which the two sides differ. The oracle reports
$(i, f)$ --- a translator (or interpreter, or carry-back) defect, localized
to a step and a named observable. Third, the projection is part of the
contract, not a tuning knob: $\proj_P$ states what ``preserves meaning''
promises for this pair, and $F_A \setminus \proj_P$ --- the pair's
\emph{loss} --- states what it does not.

The carry-back $\carry$ is what makes a target-side result \emph{mean}
something at the source: a solver counterexample or a target trace is
re-expressed as a source behavior. It is pair-specific because the
correspondence it encodes is pair-specific; it is also why the square's
oracle is as expressive as the translator it checks.

\begin{genexample}[One square, run --- and failing]
The \texttt{JAL} probe of the RV64IMC inventory
(\Cref{def:coverage}) is a two-instruction program whose taken jump
leaves the code image:
\begin{center}
\footnotesize\ttfamily
\begin{tabular}{@{}lll@{}}
0x0: & jal x1, 8 & ; link x1 := 4, jump to 0x8 \\
0x4: & ecall     & ; (jumped over) \\
\end{tabular}
\end{center}
The left leg runs the shared RISC-V interpreter $I_A$; the right leg
translates to BTOR2, interprets, and carries the trace back, compared
under $\proj \supseteq \{\texttt{pc}, \texttt{x1}, \texttt{halted}\}$.
The square is exact ($W = \mathrm{id}$) --- as is every registered
pair's but the abstraction endo-pairs' (\S\ref{sec:lax}).
Post-step states, with the right leg as the
\texttt{riscv-btor2} translator emitted them at version 0.1:
\begin{center}
\footnotesize
\begin{tabular}{@{}clll@{}}
\toprule
step & field & $I_A(p)$ & $\carry(I_B(T(p)))$ \\
\midrule
0 & \texttt{pc}, \texttt{x1}, \texttt{halted} & 8, 4, F & 8, 4, F \\
1 & \texttt{pc}, \texttt{x1} & 8, 4 & 8, 4 \\
  & \texttt{halted} & \textbf{T} & \textbf{F} \\
\bottomrule
\end{tabular}
\end{center}
Address \texttt{0x8} matches no instruction: the reference semantics
halts on the fetch miss; the translator had modeled only the
instructions. The oracle reports $(i, f) = (1, \texttt{halted})$ and,
run hop by hop (\Cref{cor:localization}), indicts the
\texttt{riscv-btor2} hop --- incident I21 (\S\ref{sec:bugs}), caught
by the first conjoined coverage run in three independently written
lowerings the same morning, fixed by a versioned translator bump with
the taken-jump squares as permanent regressions. The all-pass rows of
\Cref{tab:capability} are made of exactly such checks.
\end{genexample}

\subsection{Pasting: when squares compose}
\label{sec:pasting}

Two pairs compose when the middle language is shared:
$P_1 = (A, B, T_1, \carry_1, \proj_1, d_1)$ and
$P_2 = (B, C, T_2, \carry_2, \proj_2, d_2)$ yield the candidate composite
\[
P_2 \circ P_1 \;=\; (A,\, C,\; T_2 \circ T_1,\; \carry_1 \circ \carry_2,\;
\proj,\; d_1 \wedge d_2)
\]
for a projection $\proj \subseteq \proj_1$ to be determined, with
direction the meet on the chain $\mathsf{exact} > \mathsf{over}$ --- a
composite is exact iff both hops are --- witness embedding the
composition $W_2 \circ W_1$ (an exact hop contributing the identity),
and $\dom(P_2 \circ P_1)$ the evident composite domain. Folklore says the
squares paste. They do not paste for free: $P_2$ guarantees agreement of
the two $B$-behaviors only \emph{up to} $\proj_2$, and $\carry_1$ consumes
whole $B$-behaviors. If $\carry_1$ reads a $B$-field that $P_2$ discards,
the outer rectangle can fail while both inner squares hold. The missing
condition is that $\carry_1$ must not look outside what $P_2$ preserves:

\begin{definition}[Support]
\label{def:support}
$\carry_1$ is \emph{$(\proj_2 \Rightarrow \proj)$-supported} when for all
$b, b' \in \dom(\carry_1)$:
$\proj_2(b) = \proj_2(b')$ implies
$\proj(\carry_1(b)) = \proj(\carry_1(b'))$.
Equivalently: $\proj \circ \carry_1$ descends to a well-defined map on
$\proj_2$-equivalence classes --- a congruence condition, matching the
mechanization exactly.
\end{definition}

\begin{theorem}[Pasting]
\label{thm:pasting}
Let $p \in \dom(P_2 \circ P_1)$, let $P_1$ be faithful at $p$, let $P_2$
be faithful at $T_1(p)$, let $\proj \subseteq \proj_1$, and let
$\carry_1$ be $(\proj_2 \Rightarrow \proj)$-supported. Then the
composite $P_2 \circ P_1$ is faithful at $p$ with respect to $\proj$,
along the composed embedding $W_2 \circ W_1$ and with direction
$d_1 \wedge d_2$; for the exact--exact instance, writing
$r = T_2(T_1(p))$,
\[
\proj(I_A(p)) \;=\; \proj\bigl(\carry_1(\carry_2(I_C(r)))\bigr).
\]
\end{theorem}

\begin{proof}
Write $q = T_1(p)$. Then $\proj(I_A(p)) = \proj(\carry_1(I_B(q)))$ since
$\proj \subseteq \proj_1$ and $P_1$ is faithful at $p$. Faithfulness of
$P_2$ at $q$ gives
$\proj_2(I_B(q)) = \proj_2(\carry_2(I_C(r)))$, so the
support condition applies to the two $B$-behaviors and yields
$\proj(\carry_1(I_B(q))) = \proj(\carry_1(\carry_2(I_C(r))))$.
Chain the equalities. For the directional form, run the same chain per
closing valuation, the legs closed at $x$, $W_1(p,x)$, and
$W_2(T_1(p), W_1(p,x))$; mechanized as \texttt{lax\_pasting} and,
telescoped over routes with the direction meet,
\texttt{DRoute.lax\_route\_pasting} and
\texttt{DRoute.direction\_exact\_iff} (\S\ref{sec:mech}).
(Details: \Cref{app:proofs}.)
\end{proof}

The support condition is not decorative. Take $F_B = \{g, h\}$ with
$\proj_2 = \{g\}$, and let $\carry_1$ copy field $h$ into a
$\proj_1$-kept source field: $P_2$ can be faithful --- it preserves
$g$ --- while $\carry_2(I_C(r))$ differs from $I_B(q)$ on $h$, which
$\proj_2$ permits. The outer rectangle then fails on the copied field
although both inner squares commute.

\paragraph{Composed keep and loss}
In the implementable special case --- $\carry_1$ \emph{fieldwise} and
step-aligned, each source field $f$ computed from a dependency set
$\mathit{deps}(f) \subseteq F_B$ of target fields --- the support
condition has a syntactic reading, and the largest $\proj$ satisfying
\Cref{thm:pasting} is
\[
\keep(P_2 \circ P_1) \;=\;
  \{\, f \in \proj_1 \mid \mathit{deps}(f) \subseteq \proj_2 \,\},
\]
a source field survives the route iff $P_1$ keeps it and it is computed
only from target fields $P_2$ keeps. Dually
$\lost(P_2 \circ P_1) = F_A \setminus \keep(P_2 \circ P_1)$: route loss is
the union of per-hop losses, each pulled back along the carry-backs. This
is the formal content of the platform rule that a route must \emph{declare
its cumulative loss}: the observables a destination can still speak about
are exactly those no hop discarded. (When $\carry_1$ changes step
granularity --- one source step spanning several target steps, as in
C-to-ISA --- the same statement holds with $\mathit{deps}$ taken over the
spanned window; appendix~A.2.)

\begin{corollary}[Localization]
\label{cor:localization}
Under the support condition, if the composite square fails at $p$ (both
sides defined), then $P_1$'s square fails at $p$ or $P_2$'s square fails
at $T_1(p)$. Inductively, along a route
$R = P_n \circ \cdots \circ P_1$ satisfying the pairwise support
conditions, a failing outer rectangle implies a failing inner square at
some hop $i$ --- and running the square oracle hop by hop finds the
least such $i$, then the least failing step and field. (The check is
direction-indifferent, so this holds along any directional route ---
divergences localize the same way whether the failing square's reading
was $\equiv$ or $\sqsubseteq$.)
\end{corollary}

\Cref{cor:localization} is the debugging story of the whole platform: a
divergence anywhere along a route is never a mystery about the route; it
is an indictment of one hop, one step, one observable. It is also what
disagreement between \emph{routes} will lean on in \S\ref{sec:branching}.

\paragraph{Routes, and the contract algebra}
The registry of languages and pairs forms a directed multigraph; a
\emph{route} $R : A \to Z$ is a path through it, with composed translator
$T_R$, composed carry-back $\carry_R$, composed projection $\keep(R)$,
composed direction and embedding as above, and
faithfulness given by iterating \Cref{thm:pasting}. Composition, told
once: a pair's composable declaration is its \emph{contract} --- the
assurance class of its guarantee (\Cref{def:class}) paired with its
direction, alongside the keep-set and the measured cost --- and a
route's contract is the \emph{componentwise meet}, the weakest hop on
every axis at once, with nothing stronger entailed
(\texttt{Contract.comp\_glb}, machine-checked; the semantic content
of each coordinate is \Cref{thm:pasting} with
\Cref{prop:weakestlink}, and \Cref{prop:lax}). Purity composes too,
and gives caching:

\begin{proposition}[Determinism and caching]
\label{prop:cache}
If every component function of every hop of $R$ is pure, then $T_R$ and
$\carry_R$ are pure, and a content-addressed cache keyed by the input
hash and the component versions returns, on a hit, exactly what
recomputation would return --- across the whole route. If some hop is
impure, re-running the route on the same input and
diffing bytes at every hop (\emph{twice-and-diff}) can detect it and
localize it to the first hop whose bytes differ --- a sound but
incomplete detector: two coinciding runs of an impure hop pass.
\end{proposition}

\subsection{The direction axis: abstraction as a pair}
\label{sec:lax}

Every square so far shown ran exact, and most registered pairs do. The
$\mathsf{over}$ half of \Cref{def:faithful} earns its place on cost:
an \emph{abstraction} maps a program to a smaller or cheaper model with
\emph{more} behaviors --- cut a state variable's update and let it
range freely, and every source behavior survives while new, spurious
ones appear. Such a translation is exactly what makes an expensive
universal question tractable, and in this calculus it is not a special
mechanism bolted onto a solver run but an ordinary pair whose declared
direction is $\mathsf{over}$: same oracle, same coverage conjunction
(\Cref{def:coverage}), same ratchet (\Cref{prop:ratchet}), same
negative controls, checked along its witness embedding. Two results
pin the axis down: exactness is a special case, and universal verdicts
transfer.

\begin{proposition}[Exactness is the identity embedding]
\label{prop:exactid}
Under the specialization obligation of \Cref{thm:universal}(iv) ---
translating the closed instance agrees with closing the one open
translation --- a pair with $W = \mathrm{id}$ is faithful at an open
$p$ in the sense of \Cref{def:faithful} iff its two-sided square
commutes at every closed instance $p(x)$. The exact calculus of
\S\ref{sec:pairs}--\S\ref{sec:pasting} is thus the
$W = \mathrm{id}$ fragment of the directional one --- literally: the
mechanization derives each side from the other
(\texttt{laxFaithful\_id\_iff\_faithful}, axiom-free;
\S\ref{sec:mech}).
\end{proposition}

\begin{proof}
With $W = \mathrm{id}$ the two legs close at the same valuation, and
the specialization obligation identifies the closed translation
$T(p(x))$ with the closed instance $T(p)(x)$ of the open artifact;
the quantification over valuations is then the quantification over
closed instances.
\end{proof}

\begin{proposition}[Universal transfer across over-approximation]
\label{prop:lax}
Let $R : A \to Z$ have composed direction $\mathsf{over}$ and be
faithful (\Cref{def:faithful}) on a fragment containing $p$, with
$\varphi$ over $\keep(R)$. If no valuation of $T_R(p)$ exhibits
$\varphi$ within the declared bounds, then no valuation of $p$ does:
\emph{universal verdicts transfer across over-approximation}. An
existential verdict does not transfer; it remains subject to the
carried-back replay of \Cref{thm:existential}, whose failure on an
$\mathsf{over}$ route has a second, benign reading --- a \emph{spurious
counterexample}, a demand to refine the abstraction rather than an
indictment of a hop.
\end{proposition}

\begin{proof}
Contrapositive: a valuation $x$ with $\varphi$ observed in
$\sem{p(x)}{A}$ yields, by faithfulness along the composed embedding,
the valuation $W_R(p,x)$ of $T_R(p)$ exhibiting $\varphi$ on the kept
observables --- present within bounds, contradicting the premise.
(Mechanized: \texttt{lax\_universal\_transfer}, \S\ref{sec:mech}.)
\end{proof}

Together with \Cref{thm:existential,thm:universal} this closes the
CEGAR loop with durable artifacts: a transferred \textsc{unreachable}
answers at the source (\Cref{prop:lax}); a \textsc{reachable} either
replays (a true counterexample) or is spurious --- and the refinement
the failure demands is a \emph{new registered pair} with a smaller
havoc set, graded and reusable. The artifact instantiates this with
the endo-pair \texttt{btor2-havoc}, localization abstraction on the
bit-level hub, its lax square, negative controls, and refinement loop
in the repository's gate (a second $\mathsf{over}$ endo-pair, interval
abstraction, is registered as a brief); \S\ref{sec:eval-abstraction}
measures the loop, on authored systems and on an ingested HWMCC
slice.

\paragraph{Advising the dial}
The abstraction's parameters stay the player's, but they need not be
guessed. The artifact ships an advisor over the bit-level hub: the
\emph{cone of influence} of a question --- rooted in its bad
conditions \emph{and every environment constraint} (a state gating
run validity is never free), closed backwards through the
state-evolution supports --- splits the state space into a
\emph{free} havoc set, whose abstraction provably cannot move the
question's signal, and a refinement ladder ordered
farthest-from-the-question first; observed per-state bounds seed
candidate intervals that the lax square then corroborates or refutes.
The free-set claim is executable and negative-controlled in the
artifact's tests; mechanizing it is a natural next entry for
\S\ref{sec:mech}'s audit.

\paragraph{The dual, for symmetry}
The axis has an evident third value the platform deliberately does not
yet inhabit: an \emph{under}-approximating square --- the embedding
running the other way, the target with \emph{fewer} behaviors ---
where existential verdicts would transfer and universal ones would
not, the mirror image of \Cref{prop:lax}. Nothing in this paper
depends on it; we note it because the directional reading makes the
symmetry visible, and its verdict-transfer discipline would fall out
of the same monotonicity.

The mechanization of \S\ref{sec:mech} covers the directional calculus
in full: the directional square, pasting along the composed embedding
(binary and telescoped) with the direction meet, and both propositions
above are machine-checked (\texttt{Calculus/Lax.lean}) with the same
axiom footprint as the exact fragment.

\subsection{Contracts: grades, classes, and the meet}
\label{sec:fidelity}
\label{sec:grades}

A pair declares a \emph{fidelity grade} from
$G = \{$\ttrust, \trepr, \tchk, \tpred, \tprov$\}$ --- the
\emph{evidence species} behind its guarantee: a spec audit that
re-derives the bytes ($\tpred$), an independently re-checked
certificate ($\tprov$), this run's square-oracle verdict ($\tchk$), a
digest pin ($\trepr$), or nothing ($\ttrust$). Grades differ in kind
($\tpred$ and $\tprov$ are incomparable: foreseeable vs.\
re-checkable); what composition needs is the logical \emph{form} of
the guarantee, and that is totally ordered:

\begin{definition}[Assurance class]
\label{def:class}
The \emph{assurance classes} are the chain
$\mathbb{A} : \mathsf{none} < \mathsf{replay} < \mathsf{perrun} <
\mathsf{universal}$, and
$\mathit{class} : G \to \mathbb{A}$ maps
$\ttrust \mapsto \mathsf{none}$,
$\trepr \mapsto \mathsf{replay}$,
$\tchk \mapsto \mathsf{perrun}$, and
$\tpred, \tprov \mapsto \mathsf{universal}$.
The classes denote guarantee predicates for a pair $P$ with covered
fragment $C_P$:
$\mathsf{universal}$ --- $P$ is faithful at every $p \in C_P$;
$\mathsf{perrun}$ --- $P$ is faithful at every $p$ on which the square
oracle has run and passed (in particular, at the program of the present
run);
$\mathsf{replay}$ --- $T$ is pure, nothing about meaning;
$\mathsf{none}$ --- nothing.
\end{definition}

Only classes compose: a route with one $\tpred$ hop and one $\tprov$
hop is $\mathsf{universal}$ end to end, as it should be.

Two honesty notes on the mapping. First, grade \emph{declarations} are
themselves trusted inputs: only $\tchk$ is mechanically enforced by the
framework (the oracle runs); $\mathit{class}(\tpred) = \mathsf{universal}$
additionally presumes that the audited specification maps rules to
\emph{semantics} --- a translator faithful to a wrong rule table is
$\tpred$ and universally unfaithful, and our own evaluation exhibits a
$\tpred$ edge that was wrong on covered inputs until a coverage probe
caught it (\S\ref{sec:bugs}). Declarations therefore sit in the
trusted-base ledger like any other component. Second, $\tprov$ as
instantiated in \S\ref{sec:eval-proved} certifies a \emph{per-question}
verdict --- universal over that question's inputs, per-run over
questions --- not the pair's square for all programs.

\label{sec:composition}

\begin{proposition}[Weakest link]
\label{prop:weakestlink}
Let route $R = P_n \circ \cdots \circ P_1$ satisfy the support conditions
of \Cref{thm:pasting}. If every hop's guarantee holds, then the composite
guarantee of class $\min_i \mathit{class}(P_i)$ holds for $R$ (with
$C_R$ the composite fragment), and no higher class is entailed by the
hop guarantees alone.
\end{proposition}

The proof is immediate from \Cref{def:class} and \Cref{thm:pasting}:
universal statements conjoin into universal statements over the composite
fragment; a per-run statement at hop $i$ caps the conjunction at the
programs the oracle actually saw; a $\mathsf{replay}$ hop caps it at
purity. Optimality is by the evident counterexamples. Statically, then,
\emph{a route is as faithful as its weakest hop} --- and this is one
instance of the paper's single composition statement: a pair's
composable declaration, its \emph{contract} (assurance class $\times$
direction), composes along a route by the componentwise meet, the
weakest hop on every axis at once (mechanized:
\texttt{Contract.comp\_glb}, \S\ref{sec:mech}; the class coordinate is
this proposition, the direction coordinate is \Cref{prop:lax}). The
static story is still not the whole story, because every pair carries
its own oracle:

\begin{theorem}[Per-run re-establishment]
\label{thm:reestablish}
Fix a run of route $R$ on program $p$, with intermediate programs
$q_0 = p$, $q_i = T_i(q_{i-1})$. If on this run the square oracle runs
and passes at every hop $i$ (on $q_{i-1}$), and the support conditions
hold, then $R$ is faithful at $p$ --- regardless of the hops' static
grades. In particular a $\trepr$ or even $\ttrust$ hop that is
dynamically validated at $q_{i-1}$ contributes, for this run, exactly
what a $\tchk$ hop contributes.
\end{theorem}

\begin{proof}
Each passed oracle check \emph{is} faithfulness at $q_{i-1}$
(\Cref{def:faithful} is decidable and the oracle decides it); apply
\Cref{thm:pasting} inductively.
\end{proof}

\Cref{thm:reestablish} applies exactly to hops whose squares run: the
pinned C compiler at the spine's head has no square, so the inline
oracles lift the route \emph{from the compiled binary onward}, the
compiler contributing $\trepr$ plus a corpus-level differential
against an independent C verifier (\Cref{sec:instantiation}).
Re-establishment is per-run evidence, never a static promotion, and
never coverage for a hop with no oracle.

\subsection{Branches, diversity, and the trusted base}
\label{sec:branching}

The registry graph is not a line: one source may reach one destination by
several routes. Let $R_1, R_2 : A \to Z$ with composed keep-sets
$K_1, K_2$ and $K = K_1 \cap K_2$, sharing the endpoint interpreters
(language-owned) but no translator or carry-back.

\begin{lemma}[Disagreement localizes]
\label{lem:disagree}
If $R_1$ and $R_2$ disagree at $p$ under $K$, at least one route is
unfaithful at $p$, and \Cref{cor:localization} pins the failure to a hop,
a step, and a field of that route.
\end{lemma}

\begin{lemma}[Agreement corroborates]
\label{lem:agree}
If $R_1$ and $R_2$ agree at $p$ under $K$, then either both are faithful
at $p$ (under $K$), or both are unfaithful with \emph{identical projected
error}: defects in independently built components producing the same
wrong $K$-behavior at $p$.
\end{lemma}

Both are immediate: two faithful routes each equal $K(I_A(p))$, hence
each other. What \Cref{lem:agree} buys depends on an assumption we state
rather than hide:

\begin{assumption}[Diversity]
\label{asm:diversity}
Corroborating routes share no translator or carry-back \emph{on the
segment where they diverge} (up to the language at which they
reconverge), and the diverse translators are derived from
\emph{independent semantic artifacts} --- e.g.\ RISC-V reaches BTOR2
once via a translator written from the prose ISA specification and once
via the Sail formal model; likewise AArch64 from the Arm manual vs.\ the
Arm Sail model.
\end{assumption}

The assumption's premise is, moreover, a \emph{declared,
protected registry fact} rather than prose alone: each pair names the
semantic artifact its translator derives from, and a mechanical
judgment (the trust advisor of \S\ref{sec:books}) scores every branch
by the diverse segments' declared artifacts, the shared suffix
removed --- a shared artifact is never independent, an undeclared one
is unknown, never silently independent --- and states saturation
plainly: pairs can be generated without bound, independent artifacts
cannot, and corroboration stops growing at their supply.

The scoping matters, and independence must be argued, not assumed. Below
the reconvergence, routes may share hops --- in the platform both
RISC-V$\to$SMT-LIB routes share the BTOR2$\to$SMT-LIB bridge --- and
shared emission libraries are shared failure surface, so agreement
corroborates the diverse prefix and nothing beneath it. The evaluation
therefore also runs a \emph{disjoint-decision} variant of every branch
question (\S\ref{sec:eval-branch}, \Cref{fig:branch}): the direct route's BTOR2 system
decided natively by btormc, against the Sail-derived route through the
bridge and Z3, so the diverse segment extends through the decision
procedure and the shared bridge drops out of the corroborating pair
entirely; what remains shared is the emission library (the site of the
cross-route incident below) and the language-owned endpoints. Correlated
failures between nominally independent implementations are, moreover,
the empirical norm \citep{knight1986independence}; our own history
exhibits both a cross-route common-mode failure through a shared emitter
and a shared-misreading blind spot (\S\ref{sec:bugs}), and translators
written by LLM agents plausibly share model-family failure modes.

Under \Cref{asm:diversity}, the residual failure mode of an agreeing
branch is a \emph{common-mode failure}: independently produced encodings
of the same semantics, wrong in the same observable way on the same
program. We do not assign this a probability; we account for it
structurally, in the trusted base ledger below. This is how the platform
manufactures assurance out of hops that are individually only
$\tchk$ --- dual modular redundancy whose comparator, unlike a voter,
localizes: agreement of independently derived routes is evidence that
neither route's grade alone provides.

\label{sec:tcb}

Define the \emph{trusted computing base} $\tcb(v)$ of a verdict $v$ as
the set of components whose correctness $v$'s soundness still assumes
after all checks that ran. Each cross-check layer strictly shrinks it
(\Cref{tab:ledger}); the last two rows are not interchangeable, and the
difference is the central asymmetry of the whole design.

\begin{table}
\caption{The trusted-base ledger: what each layer of cross-checking
removes from the set of components a verdict still assumes correct.}
\label{tab:ledger}
\footnotesize
\begin{tabular}{@{}>{\raggedright\arraybackslash}p{0.36\linewidth}%
>{\raggedright\arraybackslash}p{0.25\linewidth}%
>{\raggedright\arraybackslash}p{0.30\linewidth}@{}}
\toprule
\textbf{Layer that ran} & \textbf{Removed from \tcb} &
\textbf{Still in \tcb} \\
\midrule
bare route + solver & --- &
all $T_i$, $\carry_i$, all interpreters, solver \\
+ square oracle at every hop (\Cref{thm:reestablish}) &
every $T_i$ &
the $\carry_i$, the interpreters, solver \\
+ branch agreement (\Cref{lem:agree}, \Cref{asm:diversity}) &
the diverse prefix's $\carry_i$ (cross-corroborated) &
shared endpoint interpreters, shared suffix hops and emission
libraries, solver; common-mode residue \\
+ independent solver corroboration & the single solver &
the agreeing engines' common failure modes \\
+ certificate check ($\tprov$; e.g.\ DRAT/LRAT) & the solver entirely &
the certificate checker (formally verified, in the \texttt{cake\_lpr}
instantiation), plus the bit-blaster's BV$\to$CNF step \\
+ source-level witness replay (\S\ref{sec:endtoend}) &
\emph{everything above} &
$I_A$ adequacy (\Cref{asm:adequacy}) + the replay harness \\
\bottomrule
\end{tabular}
\end{table}

\subsection{The end-to-end asymmetry}
\label{sec:endtoend}

Finally, what does a user hold at the end? Let $R : A \to Z$ be a route
into a \emph{reasoning language} $Z$ (one a solver consumes: a
transition-system logic, an SMT theory). The user asks a question about
an \emph{open} source program $p(x)$ --- free inputs $x$ --- and a
condition $\varphi$ over $\keep(R)$-observables: is $\varphi$ reachable?
The solver decides at $Z$; its answer comes back in one of two logical
shapes, and they differ fundamentally in what must be trusted.

\begin{theorem}[Existential answers are self-certifying]
\label{thm:existential}
Suppose the solver returns \textsc{reachable} with witness $w$, the
composed carry-back $\carry_R$ re-expresses $w$ as an input valuation
$x_0$ (and claimed trace), and replaying $I_A(p(x_0))$ exhibits
$\varphi$. Then $\varphi$ is truly reachable in $\sem{p}{A}$ --- with
$\tcb = \{$\Cref{asm:adequacy} for $I_A\}$ plus the small fixed replay
harness the judgment runs through (the projection, the
$\varphi$-evaluator, and the substitution closing $p$), and nothing
else. Every
translator, carry-back, hub interpreter, and solver on $R$ is a
\emph{discovery} device only: a defect in any of them can cause a missed
or non-replaying witness, never a false \textsc{reachable}.
\end{theorem}

\begin{proof}
The replay is a run of $I_A$ itself on a concrete closed program
$p(x_0)$; by \Cref{asm:adequacy} its behavior is $\sem{p(x_0)}{A}$, and
$\varphi$ was observed on it. No other component's output enters the
final judgment.
\end{proof}

\begin{theorem}[Universal answers need the machinery]
\label{thm:universal}
Suppose the solver returns \textsc{unreachable} (within declared
unrolling/step bounds $k$). If (i) every hop of $R$ is faithful at every
instance $p(x)$ --- the $\mathsf{universal}$-class guarantee on the
fragment containing $p$'s constructs, and exactly the hypothesis the
mechanization takes; (ii) the support conditions of
\Cref{thm:pasting} hold and $\varphi$ mentions only $\keep(R)$;
(iii) the verdict is certified by an independently checked certificate
or corroborated across independent engines; and (iv) the symbolic
artifact the solver consumes denotes the family of closed translations
the squares check --- translation commutes with input specialization, an
obligation each reasoning pair discharges by construction of its
encoding and, where the encoding carries free inputs, by sampled
commutation tests (\S\ref{sec:mech}), and under which the mechanization
derives every
per-instance translation from the one open translation the platform
actually performs (\texttt{universal\_from\_open\_artifact},
\S\ref{sec:mech}) --- then no input drives
$\sem{p}{A}$ to $\varphi$ within $k$ steps, with $\tcb$ as in the
corresponding row of the ledger (\S\ref{sec:tcb}).
\end{theorem}

One caution on hypothesis (i): per-run oracle evidence combined with
branch agreement (\Cref{lem:agree}, \Cref{asm:diversity}) \emph{raises
confidence} in it but does not entail it --- oracle passes cover the
finitely many closures actually run, and agreement bounds the residual
failure to common-mode rather than excluding it. Universal verdicts are
precisely where the $\mathsf{universal}$ tiers earn their keep; evidence
below that tier is corroboration, not entailment.

Direction slots into the same asymmetry monotonically: along an
$\mathsf{over}$ route a certified \textsc{unreachable} transfers to the
source (\Cref{prop:lax}), while a \textsc{reachable} already rested on
the carried-back replay --- so the directional calculus adds no new
trust obligation, only a second, benign reading of a replay failure:
a spurious counterexample, the refinement demand of \S\ref{sec:lax}.

The asymmetry between \Cref{thm:existential,thm:universal} is, we argue,
the right lens on the entire design: \emph{for questions whose answers
carry witnesses, trust is nearly free} --- carry the witness back and
replay it at the source; the platform's graded tiers, branches, and
certificates exist for the other half, the universal answers where no
witness can vouch for the verdict. This is the end-to-end argument
\citep{saltzer1984endtoend} transplanted to translation: the check that
settles a witness-carrying answer belongs at the endpoint --- the
source interpreter --- and what the intermediate hops provide is, for
that half, an optimization; only for universal answers, where no
endpoint check can exist, does hop fidelity become the guarantee
itself. A system that makes the carry-back
$\carry$ a first-class, per-pair component is a system in which the cheap
half of trust is always cashed in; a system that grades and stacks
cross-checks is one in which the expensive half is bought consciously,
in labeled increments, with the bill itemized in the
ledger~of~\S\ref{sec:tcb}.

\subsection{Mechanization}
\label{sec:mech}

The compositional core of \Cref{sec:calculus,sec:fidelity} is
mechanized in Lean~4 (core library only, no mathlib; $\sim$1,250 lines,
in the artifact): the pasting theorem with the support condition,
localization, the assurance-class chain and weakest link (sound half),
per-run re-establishment, both branch lemmas, the ratchet (verdict
preservation and coverage monotonicity under extensions), both
end-to-end theorems, the \emph{specialization obligation} of
\Cref{thm:universal}(iv) --- with the theorem that it discharges the
per-instance-translation hypothesis from the single open translation
the platform performs --- and the \emph{route telescope} --- routes of
arbitrary length as a dependently-typed chain of pairs over
interpreter-bundled languages, with the paper's ``pairwise support
conditions'' as a recursive coherence predicate, and the $n$-ary
pasting and localization theorems proved by induction through the
binary ones (the canonical head-projection choice at each junction is
shown lossless by a reprojection lemma). Three modeling
choices mirror the paper exactly. Components are Lean functions, so
\emph{purity is by construction} and \Cref{prop:cache} is
definitional; partiality is \texttt{Option}, with \unsupported{} as
\texttt{none}; and projections are generalized from field subsets to
arbitrary observable-compression maps, with ``$\proj \subseteq
\proj_1$'' becoming \emph{factoring}. Because every theorem takes each
component-output witness as an explicit hypothesis, a statement's
trusted base is literally its hypothesis list ---
\Cref{thm:existential} in Lean assumes source-interpreter adequacy and
nothing else, and its proof term is three tokens:
\begin{lstlisting}[basicstyle=\ttfamily\footnotesize,
  literate={∀}{{$\forall$}}1 {→}{{$\to$}}1 {₀}{{$_0$}}1
           {φ}{{$\varphi$}}1 {π}{{$\pi$}}1 {∃}{{$\exists$}}1
           {∧}{{$\land$}}1 {⟨}{{$\langle$}}1 {⟩}{{$\rangle$}}1]
theorem existential_self_certifying
  (adequacy : ∀ p b, IA p = some b → sem p = some b)
  (hreplay : IA p₀ = some b) (hφ : φ (projB π b)) :
  ∃ b', sem p₀ = some b' ∧ φ (projB π b') :=
  ⟨b, adequacy _ _ hreplay, hφ⟩
\end{lstlisting}
\Cref{tab:mech} maps the paper's results to their Lean theorems with
each proof's exact axiom footprint, re-checked at every build; the one
classical proof is \Cref{cor:localization} (an unfaithful route names
no witness by itself). The table also shows where mechanization stops,
and we state it plainly: \Cref{thm:universal}'s clauses (iii) and (iv)
enter the Lean statement as \emph{hypotheses}. Both now carry
\emph{tested surrogates} rather than by-construction assertions. The
specialization obligation is discharged by sampled commutation checks
per reasoning pair with free inputs (in the suite): the open program is
translated once, and for sampled valuations the emitted artifact is
byte-identical while the specialized instance's square passes --- the
shape \texttt{universal\_from\_open\_artifact} consumes. The
solver-artifact-to-target-semantics correspondence behind (iii) is
corroborated per verdict: every bounded-unreachable route verdict in
the branch, benchmark, and certified experiments is additionally
replayed through the strict
target interpreter for the full bound --- sampled inputs where the
system has any --- and no \texttt{bad} may fire
(\S\ref{sec:eval-branch}--\S\ref{sec:eval-proved}). Sampling
corroborates, it does not entail --- which is the grade discipline
applied to the calculus's own hypotheses.

\begin{table}
\caption{The mechanization at a glance: paper result, Lean theorem,
and the proof's axiom footprint as printed by the build-time audit
(\texttt{Audit.lean}); ``---'' is axiom-free.}
\label{tab:mech}
\footnotesize
\begin{tabular}{@{}>{\raggedright\arraybackslash}p{0.24\linewidth}%
>{\raggedright\arraybackslash}p{0.45\linewidth}%
>{\raggedright\arraybackslash}p{0.22\linewidth}@{}}
\toprule
\textbf{Result} & \textbf{Lean theorem} & \textbf{Axioms} \\
\midrule
\Cref{thm:pasting} (pasting) & \texttt{pasting}, \texttt{pasting}$_3$
  & \texttt{propext} \\
\Cref{cor:localization} & \texttt{localization}
  & \texttt{propext}, \texttt{Classical.\allowbreak choice},
  \texttt{Quot.sound} \\
\Cref{prop:weakestlink} (sound half) & \texttt{weakest\_link\_universal}
  & \texttt{propext} \\
\Cref{thm:reestablish} & \texttt{reestablishment} & \texttt{propext} \\
\Cref{lem:disagree,lem:agree} & \texttt{disagreement\_localizes},
  \texttt{agreement\_corroborates} & --- \\
\Cref{prop:ratchet} & \texttt{ratchet\_preserves\_faithful},
  \texttt{ratchet\_coverage\_mono} & --- \\
\Cref{thm:existential} & \texttt{existential\_self\_certifying} & --- \\
\Cref{thm:universal} & \texttt{universal\_needs\_machinery}
  & \texttt{propext} \\
\Cref{thm:universal}(iv) & \texttt{universal\_from\_}\newline
  \texttt{open\_artifact} & \texttt{propext} \\
route telescope ($n$-ary) & \texttt{Route.\allowbreak route\_pasting},
  \texttt{Route.\allowbreak route\_localization}
  & \texttt{propext}, \texttt{Quot.sound}
  (+ \texttt{Classical.\allowbreak choice} for localization) \\
directional pasting, direction meet (\Cref{thm:pasting}) &
  \texttt{lax\_pasting}; telescoped:
  \texttt{DRoute.\allowbreak lax\_route\_pasting},
  \texttt{DRoute.\allowbreak direction\_exact\_iff} & \texttt{propext}
  (+ \texttt{Quot.sound} telescoped) \\
\Cref{prop:exactid} (exactness $=$ identity embedding) &
  \texttt{laxFaithful\_id\_iff\_faithful} & --- \\
\Cref{prop:lax} (universal transfer) &
  \texttt{lax\_universal\_transfer} & \texttt{propext} \\
contract algebra (composition = componentwise meet) &
  \texttt{Contract.comp\_glb}, \texttt{Direction.comp\_eq\_min}
  & \texttt{propext}, \texttt{Quot.sound}
  (\texttt{comp\_eq\_min}: ---) \\
\bottomrule
\end{tabular}
\end{table} Deliberately not mechanized: the fieldwise loss
computation and retiming case of appendix~A.2\ (syntactic
side
conditions over a concrete carry-back representation), the optimality
halves of \Cref{prop:weakestlink} (meta-statements over models), and
the grade set $G$ itself --- evidence provenance is not a mathematical
object, which is the point of separating $G$ from $\mathbb{A}$. The
telescope theorems' audit adds only \texttt{Quot.sound} (from
structural-recursion equations); localization remains the sole
classical proof among these results (the same build audit also prints
the artifact's frontier-exploration model --- no result of this
paper --- whose three search lemmas are classical).


\section{The platform}
\label{sec:instantiation}

The calculus is implemented as \hg, a Python platform whose structure
mirrors \Cref{sec:calculus} exactly: a \emph{framework} that holds no
pair semantics, per-language \emph{shared interpreters}, and per-pair
translators and carry-backs. \Cref{fig:registry} draws the registry
graph; this section describes what the use plane runs on,
\Cref{sec:economy} describes how it grows, and \Cref{sec:evaluation}
measures both.

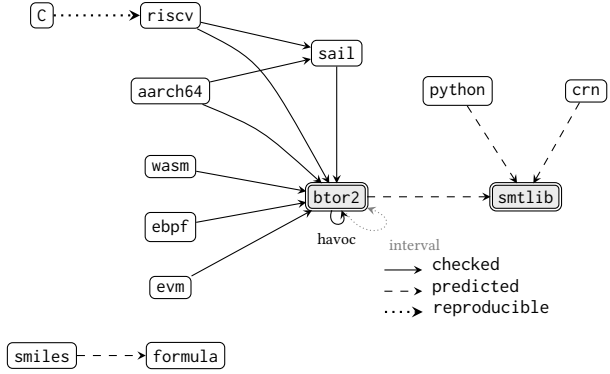
\begin{figure}
\centering
\begin{tikzpicture}[x=1cm, y=1cm,
  reg/.style={->, >=stealth, gray, densely dotted}]
\node[lang] (c)     at (0.0, 0.0)  {C};
\node[lang] (riscv) at (1.7, 0.0)  {riscv};
\node[lang] (a64)   at (1.7, -1.0) {aarch64};
\node[lang] (wasm)  at (1.7, -2.0) {wasm};
\node[lang] (ebpf)  at (1.7, -2.8) {ebpf};
\node[lang] (evm)   at (1.7, -3.6) {evm};
\node[lang] (sail)  at (3.9, -0.5) {sail};
\node[hub]  (btor)  at (3.9, -2.4) {btor2};
\node[hub]  (smt)   at (6.4, -2.4) {smtlib};
\node[lang] (py)    at (5.5, -1.0) {python};
\node[lang] (crn)   at (7.2, -1.0) {crn};
\node[lang] (smi)   at (0.0, -4.5) {smiles};
\node[lang] (form)  at (1.9, -4.5) {formula};
\draw[repr] (c) -- (riscv);
\draw[chk] (riscv) -- (sail);
\draw[chk] (a64) -- (sail);
\draw[chk] (sail) -- (btor);
\draw[chk] (riscv) .. controls (2.9, -0.6) .. (btor);
\draw[chk] (a64) .. controls (2.9, -1.5) .. (btor);
\draw[chk] (wasm) -- (btor);
\draw[chk] (ebpf) -- (btor);
\draw[chk] (evm) -- (btor);
\draw[pred] (btor) -- (smt);
\draw[pred] (py) -- (smt);
\draw[pred] (crn) -- (smt);
\draw[pred] (smi) -- (form);
\draw[chk] (btor) to[out=250, in=290, looseness=5]
  node[below, font=\tiny] {havoc} (btor);
\draw[reg] (btor) to[out=300, in=340, looseness=5]
  node[below right, font=\tiny] {interval} (btor);
\node[font=\scriptsize, anchor=west] at (4.4, -3.6)
  {\begin{tabular}{@{}l@{}}
   \tikz\draw[chk] (0,0)--(0.5,0); \ \tchk\\
   \tikz\draw[pred] (0,0)--(0.5,0); \ \tpred\\
   \tikz\draw[repr] (0,0)--(0.5,0); \ \trepr
   \end{tabular}};
\end{tikzpicture}
\caption{The registry: 13 languages (nodes; the two reasoning hubs
shaded) and 15 registered pairs (edges), each carrying its declared
grade. The two loops on the BTOR2 hub are the abstraction endo-pairs of
the direction axis (\S\ref{sec:lax}): \texttt{btor2-havoc}, built, and
--- gray --- \texttt{btor2-interval}, registered as a brief. The
evaluation of \Cref{sec:evaluation} covers the thirteen pairs built at
its snapshot. RISC-V and AArch64 each reach BTOR2 twice ---
directly (manual-derived) and through Sail (formal-model-derived) ---
the independently derived branches \Cref{asm:diversity} requires, each
leg's semantic artifact declared per pair; the
dashed \texttt{btor2}$\to$\texttt{smtlib} bridge is the shared suffix
that the disjoint-decision runs of \S\ref{sec:eval-branch} remove
from the corroborating pair.}
\label{fig:registry}
\end{figure}

\subsection{Framework}
The framework provides the registry, the content-addressed cache of
\Cref{prop:cache}, the generic square oracle with step/observable
localization, the route enumerator and runner (routes are enumerated
and reported — composed grade, coverage, loss — never chosen), the
coverage harness, the merge-run route grader, the solver/checker
plumbing, and the books of \S\ref{sec:books} with their advisory
reads (route report, answerability diagnosis, trust and reduction
advisors) — all of which annotate and account, none of which chooses.
The same surface is the \emph{player interface}: an MCP server
(\texttt{gurdy mcp}, stdio JSON-RPC) exposes exactly the enumerated,
advisory calls — languages, pairs, annotated routes, coverage, the
diagnosis, the advisors, and bounded decide-with-replay — so an LLM
player connects to the platform as a machine peer, with the exposure
rule pinned by test (\S\ref{sec:eval-player2} runs the player
experiment over this surface, cold).
Everything except the solver call is byte-deterministic, enforced
repo-wide by twice-and-diff.

\subsection{Languages and reasoning hubs}
Thirteen languages are registered, each with a specification-anchored
formal semantics and (for all but C) a built shared
interpreter: the ISAs RISC-V (RV64IMC), AArch64 (a growing A64 slice),
eBPF, Wasm, and EVM; C (whose ``interpreter'' role is played
differentially, see below); Python (a pinned CPython restricted by an AST
allow-list to an integer subset --- the loader rejects, CPython runs);
Sail (an executable slice of the Sail-derived RISC-V and Arm semantics);
chemical reaction networks (a discrete Petri-net stepper); SMILES
molecular graphs and molecular formulas; and the two \emph{reasoning
hubs}: BTOR2, a bit-vector/array transition-system logic consumed by
hardware model checkers --- environment \texttt{constraint}s enforced
with the per-frame reading, a bad counting only on a constraint-valid
prefix, on which the shared evaluator, the bridged and the native
engines agree --- and SMT-LIB ($\mathrm{QF\_ABV}$ and
$\mathrm{QF\_LIA}$ fragments).

Reasoning languages own, beyond an interpreter, a solver inventory
(BtorMC and Pono for BTOR2; Z3, Bitwuzla, Boolector, cvc5, Yices2 for
SMT-LIB) and a witness-checker inventory: interpreter replay of BTOR2
\texttt{.wit} witnesses, model evaluation by the deterministic SMT-LIB
evaluator, multi-engine corroboration, and a bit-blasted certificate
pipeline for \tprov-tier unreachability with two checker rungs: the
DRAT proof checked by \texttt{drat-trim} (independent, unverified),
and --- preferred when present --- elaborated to LRAT
(\texttt{drat-trim} as an \emph{untrusted} elaborator) and re-validated
by \texttt{cake\_lpr}, the \emph{formally verified} CakeML checker,
built natively from its compiler-generated ARMv8 assembly
(\S\ref{sec:eval-proved}).

Interpreter adequacy (\Cref{asm:adequacy}) is discharged differentially:
the RISC-V interpreter against the Sail-generated reference simulator
(\texttt{sail\_riscv\_sim}) and the official \texttt{riscv-tests} /
\texttt{riscv-arch-test} suites; the Python subset against pinned
CPython by construction --- a circularity we note plainly: for this one
language the pinned runtime is both the semantics and the interpreter,
so adequacy has no independent reference; C against CBMC on the same
source (with a
classification step separating C-undefined-but-ISA-defined behavior from
genuine faults).

\subsection{Pairs and routes}
Fifteen pairs are registered. Thirteen --- built to varying,
\emph{measured} extents (\Cref{tab:capability}) --- form the evaluation
snapshot of \Cref{sec:evaluation}; the other two sit on the direction
axis of \S\ref{sec:lax}: \texttt{btor2-havoc}, the localization
abstraction, built, negative-controlled, and driving the refinement
demonstration, and \texttt{btor2-interval}, a second over-approximation
registered as a brief under the discipline of \S\ref{sec:books}. The
spine is
C $\to$ RISC-V $\to$ BTOR2 $\to$ SMT-LIB, with the C head a digest-pinned
gcc at grade $\trepr$: the square-verified portion of the spine begins
at the compiled binary (\S\ref{sec:composition}), and the C hop is
re-checked separately by the CBMC differential above. The direct
RISC-V$\to$BTOR2 lowering draws on the rotor bounded model checker
\citep{rotor2017basics}. RISC-V and AArch64 each reach BTOR2 by two
independently derived routes (manual-derived direct translator vs.\
Sail-model-mediated), instantiating \Cref{asm:diversity}; the BTOR2 hub
funnels five ISAs into one solver interface, and the
$\mathrm{BTOR2}\to\mathrm{SMT}$-LIB bridge (a rule-for-rule operator
mapping at grade $\tpred$) connects the hubs, so every BTOR2 question can
also be decided natively-vs-bridged --- solve-step corroboration. CRN and
Python reach SMT-LIB directly ($\mathrm{QF\_LIA}$), and
SMILES$\to$formula exercises the calculus away from computation entirely.

\section{The economy: how the platform grows}
\label{sec:economy}

The evolution plane writes what \Cref{sec:instantiation} reads. Its
stages — build under contract, gate, ratchet, and the books that
direct them — are the growth step of \Cref{sec:overview}, kept under
one discipline: \emph{growth never answers}.

\subsection{How the pairs were built}
\label{sec:agents}
Each pair (and each shared interpreter) was implemented by an
\emph{independent LLM agent} from a one-page registration brief ---
under the current discipline, a brief that opens by citing the demand
records recommending the pair (\S\ref{sec:books}) --- against
the calculus as a written contract, under three standing rules: agents
never modify the framework or another deliverable's code; interpreter
changes are versioned events that re-validate every dependent pair
(\Cref{prop:ratchet}); and all growth is by additive widening --- extend
the covered fragment, re-run all prior evidence byte-for-byte, ratchet
the coverage number. The platform's checks --- the square oracle, the
differentials, solver corroboration, branch agreement, twice-and-diff ---
were the only \emph{semantic} quality gate; no human reviewed agent code
for semantic correctness, though an autonomy protocol did escalate one
class of change --- edits to shared emission code --- for human sign-off
(``largely unsupervised'' means exactly this: the human gated whether
such a change could land, never whether it was right).
The defects those checks caught, and where they localized,
are part of the evaluation (\S\ref{sec:bugs}): the implementation process
itself is an experiment in the paper's thesis that translation trust can
be manufactured by architecture rather than by trusting the translator's
author --- whether that author is a compiler vendor or a language model.
The same discipline is the growth model: pairs contributed by anyone
--- with LLMs, with agents, or by hand --- through ordinary pull
requests, admitted by architecture rather than authorship. A
contribution arrives as the calculus's tuple with its evidence surface
attached (translator, carry-back, declared projection and grade, typed
\unsupported{} partiality, probes, a runnable square); the merge gate
re-runs conjoined coverage, dependent branch agreements, and
twice-and-diff --- the machinery of \Cref{sec:evaluation} applied as
an admission test --- and \Cref{prop:ratchet} guarantees an additive
contribution extends, never silently invalidates, the verdicts already
earned.

\subsection{Coverage, and the ratchet}
\label{sec:coverage}

Fidelity alone is gameable by triviality: a pair handling one instruction
is vacuously $\tprov$. Its companion axis is measured against a yardstick
the implementer does not choose:

\begin{definition}[Coverage]
\label{def:coverage}
Fix for each language a \emph{construct inventory} $K_A$ enumerated from
its specification (opcodes, operators, syntactic forms). A construct
$k \in K_A$ \emph{counts} for pair $P$ iff $P$ covers it
($k$'s probe programs lie in $\dom(P)$) \emph{and} $P$ is faithful on
those probes: $\cov(P) \subseteq K_A$ is the per-construct
\emph{conjunction} of covered and faithful. Route coverage $\cov(R)$
counts $k$ iff its lowering survives --- covered and faithful --- at every
hop.
\end{definition}

Coverage is gamed by unsoundness (accept everything, translate it
wrongly) exactly as fidelity is gamed by triviality (translate one thing,
prove it); the conjunction closes both routes to a vacuous ``pass.''
Everything a pair rejects must be a \emph{typed} $\unsupported$, so the
uncovered set is itself enumerable and auditable --- partiality is
fail-fast, never silent.

Coverage then grows by a discipline with a soundness theorem:

\begin{proposition}[Ratchet]
\label{prop:ratchet}
Call an update $I' \sqsupseteq I$ (likewise $T$, $\carry$) an
\emph{extension} when $\dom(I') \supseteq \dom(I)$ and
$I'\!\restriction_{\dom(I)} = I$. Under extensions: every faithfulness
verdict at every $p \in \dom(I)$ is preserved; $\cov$ is monotone; and
prior evidence (oracle passes, differentials, agreement records) never
needs re-earning. An update that is \emph{not} an extension can silently
invalidate any of these, and must instead bump the component's version
and re-validate every dependent pair.
\end{proposition}

The implementable extension check is byte-identity of all old outputs
across the update --- the same twice-and-diff mechanism as
\Cref{prop:cache}, run across versions. This is precisely the repo's
widening protocol (additive interpreter bumps, versioned events,
dependent-pair re-validation), which \Cref{sec:evaluation} measures over
the project's own history.

The ratchet is the evolution plane's contract with the use plane:
declarations only strengthen, verdicts only accumulate. What directs
the growth is the demand the use plane records:

\subsection{Keeping the books: demand, cost, and recommendation}
\label{sec:books}

Between diagnosis and registration sits the accounting: the
\emph{books}, one append-only, opt-in, host-local ledger with two
sides.
The \emph{cost} side is written beside the work by the instrumented call
sites --- translation on cache miss, the square oracle, every decide
backend --- and accumulates from the capped route-grader runs in CI; it
feeds the measured axis of the annotated route report, where
Pareto-dominated routes are marked, never hidden, dominance is computed
only between fully measured routes, and an absent measurement reads
\emph{unmeasured}, never zero. The \emph{demand} side is written by the
diagnosis calls: each question the platform cannot satisfy is recorded
verbatim with its failing obstacle and the generation target it names
--- the four obstacles of \S\ref{sec:overview}, plus a fifth,
\emph{trust}, when a player's assurance floor is unmet --- and every
record carries an \emph{origin} tag separating organic player sessions
from synthetic campaigns.

A recommendation is then an aggregation, not a judgment: the board
groups recorded demand per target and deduplicates by question
identity; the failing obstacle is the platform's one demand taxonomy,
naming what the candidate pair pays for --- connectivity and shape buy
capability, loss buys coverage, cost buys performance, trust buys an
independent anchor. A
campaign of five questions, run against the July 2026 registry, yields
a four-row board that illustrates the shape of the evidence (the
diagnosis and the board are measured as a benchmark in
\S\ref{sec:eval-campaign}): a
reasoning-language demand (two liveness questions no registered hub
expresses --- the depth growth of the conclusion), an independent-pair
demand (a single-route source whose provenance is undeclared), a
connectivity demand (a source with no route to any hub, with a draft
brief stub attached), and a cost demand (a \emph{resource-out}
verdict naming the registered abstraction dials). Each row removes a
different obstacle; no scalar ranks them.

The discipline this buys is \emph{recommended, then registered}: a
registration brief opens by citing the demand records behind it --- no
evidence, no brief --- while registration itself remains the human act
it has always been. The untrusted-generator stance of
\S\ref{sec:agents} extends from building pairs to \emph{choosing} them:
an LLM whose campaign manufactures demand to justify a pet pair cannot
launder it, because demand is auditable (questions verbatim, origins
displayed apart, volume explicitly not a verdict), and the trust-side
demands rest on the declared, protected semantic artifacts of
\Cref{asm:diversity} --- a shared artifact is never independent, an
undeclared one never silently so.

\section{Evaluation}
\ifarxiv
\noindent The fifteen subsections below answer four questions. \emph{Does
the graph answer correctly?} --- the capability snapshot, composed
coverage and branch agreement, the compliance slice, and the end-to-end
case studies (\S\ref{sec:eval-capability}--\S\ref{sec:eval-cases}).
\emph{Does the evidence mean what it claims?} --- the certified tier
and the escape-rate experiments for the gate itself
(\S\ref{sec:eval-proved}, \S\ref{sec:eval-escape}). \emph{Can
untrusted authors build it?} --- the defects the architecture caught
(\S\ref{sec:bugs}). \emph{Can an LLM play it?} --- the player
experiment (\S\ref{sec:eval-player}); performance and scale close
(\S\ref{sec:eval-perf}). The snapshot tables end there; five later
benchmark families, each stamped with its own commit, follow
(\S\ref{sec:eval-constraint}--\S\ref{sec:eval-player2}), and a
short record of what remains narrated rather than measured ends the
section (\S\ref{sec:eval-exhibits}).
\fi
\label{sec:evaluation}


We evaluate the snapshot along the axes the calculus makes measurable:
per-pair coverage (\S\ref{sec:eval-capability}), composed route
coverage and branch agreement --- including a disjoint-decision block
--- (\S\ref{sec:eval-branch}), a
compliance-derived reachability benchmark decided along both RISC-V
routes (\S\ref{sec:eval-bench}), end-to-end case
studies exercising \Cref{thm:existential} (\S\ref{sec:eval-cases}),
certified-unreachability runs inhabiting the \tprov\ tier
(\S\ref{sec:eval-proved}), the defects the architecture caught
(\S\ref{sec:bugs}), a measured escape-rate estimate for the gate
(\S\ref{sec:eval-escape}), a first controlled experiment with LLM
players (\S\ref{sec:eval-player}), and performance, determinism, and
scale (\S\ref{sec:eval-perf}). Every number in this section's tables is
regenerated by one script from
the live registry and real runs at the snapshot commit, with a few
curated exceptions marked at source --- the defect catalog behind
\Cref{tab:bugs} is hand-mined from history, and \Cref{tab:player}
formats the recorded outputs of a manual two-arm protocol; the escape
experiment's earlier-round numbers (\S\ref{sec:eval-escape}) likewise
come from its incident record. (Host:
macOS/arm64, CPython 3.12, Z3 4.13, btormc 3.2, cadical; drat-trim
pinned at source commit \texttt{2e3b2dc} and cake\_lpr at
\texttt{a4323b2}, the engine inventory recorded in the regenerated
\texttt{env.json}; the platform's own suite runs 1215 tests there, all
green, with 3 host-dependent skips.) The evaluation remains deliberately
\emph{fully specified} rather than large: every question --- authored
or derived --- is checked end to end with its ground truth stated, and
the largest sweep (the 78-question compliance-derived benchmark of
\S\ref{sec:eval-bench}) is still small next to an SV-COMP-scale
campaign, which remains future work.

\subsection{Capability snapshot}
\label{sec:eval-capability}

\begin{table}
\caption{Capability snapshot: per-pair construct
coverage against the language-owned inventory. \emph{Accepted} =
the probe translates without a typed \unsupported{} abort;
\emph{conjoined} = accepted \emph{and} the pair's square oracle
passes on the probe --- Definition~\ref{def:coverage}'s conjunction,
measured directly (\S\ref{sec:eval-capability}). Pairs without a
decidable square (the \texttt{predicted}-grade hops into the
SMT-LIB hub) discharge the faithfulness conjunct per run
instead (\S\ref{sec:composition}), marked \emph{per-run}; the
reproducible C hop carries no construct inventory and shows
``---''. Gaps counts
the distinct typed \unsupported{} constructs. Measured from the
live registry at the snapshot commit.}
\label{tab:capability}
\footnotesize
\begin{tabular}{@{}llllr@{}}
\toprule
Pair (source$\to$target) & Grade & Accepted & Conjoined & Gaps \\
\midrule
c-riscv & \texttt{reproducible} & --- & --- & --- \\
riscv-btor2 & \texttt{checked} & 96/96 & 96/96 & 0 \\
riscv-sail & \texttt{checked} & 96/96 & 96/96 & 0 \\
sail-btor2 & \texttt{checked} & 95/95 & 95/95 & 0 \\
aarch64-btor2 & \texttt{checked} & 27/33 & 27/33 & 6 \\
aarch64-sail & \texttt{checked} & 27/33 & 27/33 & 6 \\
wasm-btor2 & \texttt{checked} & 54/75 & 54/75 & 21 \\
ebpf-btor2 & \texttt{checked} & 126/126 & 126/126 & 0 \\
evm-btor2 & \texttt{checked} & 86/144 & 86/144 & 58 \\
btor2-smtlib & \texttt{predicted} & 56/56 & per-run & 0 \\
crn-smtlib & \texttt{predicted} & 10/10 & per-run & 0 \\
python-smtlib & \texttt{predicted} & 11/27 & per-run & 16 \\
smiles-formula & \texttt{predicted} & 14/17 & 14/17 & 3 \\
\bottomrule
\end{tabular}
\end{table}

\Cref{tab:capability} reports per-construct coverage measured at
harvest time --- not transcribed from documentation --- in both of
\Cref{def:coverage}'s readings: \emph{accepted} (the probe translates
without a typed $\unsupported$ abort) and \emph{conjoined} (accepted
\emph{and} the pair's square oracle passes on the probe --- the
definition's actual conjunction, run probe by probe). For every pair
with a decidable square the two columns now coincide; they did not
when the conjunction was first measured, and what it surfaced on its
first run is instructive enough to report
(\S\ref{sec:bugs}, incidents I20--I22): the Sail-route front dropped
the program's \emph{initial memory}, so every load-family probe
diverged at step 0 (accepted, silently wrong --- precisely the failure
mode acceptance counting cannot see), and three of the BTOR2 lowerings
modeled a program counter that leaves the code as \emph{stuck} where
every reference interpreter halts, so every taken-jump probe whose
target exits the image diverged on \texttt{halted} (the worked example
of \S\ref{sec:pairs} shows this square, and its oracle output, verbatim). It also audited its
own instrument: one probe (\texttt{SD}) stored over its own
\texttt{ECALL} terminator, so the reference interpreter could not run
it and no square had ever executed on that construct --- acceptance
counted it covered anyway. Measuring the
definition's conjunction found all three in its first run; the fixes
are ordinary versioned translator bumps and one probe repair. Pairs without a decidable square
--- the \texttt{predicted}-grade hops into the SMT-LIB hub --- discharge
the faithfulness
conjunct per run at question time instead (\S\ref{sec:composition}) ---
what their grade leaves to per-run checking; the table marks them
\emph{per-run} (the fourth \texttt{predicted} pair,
SMILES$\to$formula, has executable ends and conjoins directly). Denominators are language-owned inventories
(\Cref{def:coverage}'s yardstick): both RISC-V-headed pairs measure
against the same 96-construct RV64IMC inventory, the two AArch64 pairs
against one declared A64 slice whose out-of-scope constructs count in
the denominator, and the Wasm and EVM rows against their in-scope
$\cup$ enumerated-out-of-scope unions (54/75, 86/144) --- a gap shows
as a typed $\unsupported$ entry, never by shrinking the total. Depth
remains uneven by design and visibly so: the spine is deep (all 96
RV64IMC constructs conjoined; eBPF complete at 126; the
BTOR2$\to$SMT-LIB bridge total on its 56-item
operator/sort/directive inventory), Python stands at 11 of 27 subset
constructs, and every gap is a \emph{typed} $\unsupported$ the harness
can enumerate (the ``gaps'' column). The declared A64 slice is still
narrow against the full base ISA --- the honest reading of 27/33 is
``the in-scope family conjoins, and the slice is small.'' The C head
carries no construct inventory: its translator is the pinned compiler,
graded $\trepr$ (\S\ref{sec:composition}).

\subsection{Composed coverage and branch agreement}
\label{sec:eval-branch}

\begin{table}
\caption{Composed \emph{conjoined} coverage: a
probe counts iff it survives every hop \emph{and} every hop with a
decidable square passes it on that hop's input (the route-level
reading of Definition~\ref{def:coverage}; ``via Sail'' routes are
the independently derived branch; denominators are the source
language's inventory, \S\ref{sec:eval-branch}). Routes marked
$^\dagger$ contain no decidable-square hop at all (their one
translator is \texttt{predicted}-grade): the number is acceptance,
with faithfulness discharged per run (\S\ref{sec:composition}).
Branch
agreement: the same question decided along both routes. Times are
the slower route, end to end (translate every hop + decide with Z3).
The bottom block decides the same questions with fully disjoint
decision stacks after the head: the direct route's BTOR2 system decided
natively by btormc (no bridge, no SMT-LIB, no Z3) against the
via-Sail route through the bridge and Z3 --- the diverse segment
spans both the lowering derivation and the decision procedure
(\S\ref{sec:branching}).}
\label{tab:branch}
\footnotesize
\begin{tabular}{@{}lr@{}}
\toprule
Route & Composed coverage \\
\midrule
RISC-V, direct & 96/96 \\
RISC-V, via Sail & 96/96 \\
AArch64, direct & 27/33 \\
AArch64, via Sail & 27/33 \\
Wasm & 54/75 \\
eBPF & 126/126 \\
EVM & 86/144 \\
CRN$^\dagger$ & 10/10 \\
Python$^\dagger$ & 11/27 \\
SMILES & 14/17 \\
\bottomrule
\end{tabular}

\medskip

\begin{tabular}{@{}lclr@{}}
\toprule
Question (both routes) & Agree & Verdict & Time (s) \\
\midrule
riscv const x1==42 (reach) & \checkmark & reach & 0.0 \\
riscv const x1==99 (unreach) & \checkmark & unreach & 0.0 \\
riscv loop sum==15 (reach) & \checkmark & reach & 0.0 \\
riscv loop sum==99 (unreach) & \checkmark & unreach & 0.0 \\
riscv store/load 0x123 (reach) & \checkmark & reach & 0.0 \\
riscv store/load 0x999 (unreach) & \checkmark & unreach & 0.0 \\
C a0==47 (reach) & \checkmark & reach & 0.1 \\
C a0==99 (unreach) & \checkmark & unreach & 0.0 \\
aarch64 movz/add x1==42 (reach) & \checkmark & reach & 0.0 \\
aarch64 movz/add x1==999 (unreach) & \checkmark & unreach & 0.0 \\
aarch64 SUBS/B.NE loop x0==1 (reach) & \checkmark & reach & 0.0 \\
aarch64 SUBS/B.NE loop x0==5 (unreach) & \checkmark & unreach & 0.0 \\
aarch64 flags + loop (SUBS/B.NE) & \checkmark & trace $=_\pi$ & --- \\
aarch64 32-bit W forms & \checkmark & trace $=_\pi$ & --- \\
\bottomrule
\end{tabular}

\medskip

\begin{tabular}{@{}lcl@{}}
\toprule
Question, decision stacks fully disjoint & Agree & Verdict \\
\midrule
riscv const x1==42 (reach) & \checkmark & reach \\
riscv const x1==99 (unreach) & \checkmark & unreach \\
riscv loop sum==15 (reach) & \checkmark & reach \\
riscv loop sum==99 (unreach) & \checkmark & unreach \\
riscv store/load 0x123 (reach) & \checkmark & reach \\
riscv store/load 0x999 (unreach) & \checkmark & unreach \\
aarch64 movz/add x1==42 (reach) & \checkmark & reach \\
aarch64 movz/add x1==999 (unreach) & \checkmark & unreach \\
aarch64 SUBS/B.NE loop x0==1 (reach) & \checkmark & reach \\
aarch64 SUBS/B.NE loop x0==5 (unreach) & \checkmark & unreach \\
C a0==47 (reach) & \checkmark & reach \\
C a0==99 (unreach) & \checkmark & unreach \\
\bottomrule
\end{tabular}
\end{table}

\begin{figure}
\centering
\begin{tikzpicture}[x=1cm, y=1cm]
\node[art]  (p)    at (3.4, 0)     {RISC-V program $p$ + question};
\node[lang] (b1)   at (1.4, -1.0)  {btor2};
\node[lang] (sail) at (5.4, -1.0)  {sail};
\node[lang] (b2)   at (5.4, -2.0)  {btor2};
\node[hub]  (smt)  at (5.4, -3.0)  {smtlib};
\node[tool] (mc)   at (1.4, -4.0)  {btormc};
\node[tool] (z3)   at (5.4, -4.0)  {z3};
\node[art]  (v1)   at (1.4, -5.0)  {verdict};
\node[art]  (v2)   at (5.4, -5.0)  {verdict};
\draw[chk]  (p) -- node[elbl, left=2pt]  {riscv-btor2} (b1);
\draw[chk]  (p) -- node[elbl, right=2pt] {riscv-sail}  (sail);
\draw[chk]  (sail) -- node[elbl, right]  {sail-btor2}  (b2);
\draw[pred] (b2) -- node[elbl, left]     {btor2-smtlib} (smt);
\draw[chk]  (b1) -- node[elbl, right=2pt, align=left, font=\tiny\itshape]
  {decide natively\\(no bridge, no\\SMT-LIB, no Z3)} (mc);
\draw[chk]  (smt) -- (z3);
\draw[chk]  (mc) -- (v1);
\draw[chk]  (z3) -- (v2);
\draw[<->, >=stealth] (v1) -- node[elbl, above] {must agree} (v2);
\end{tikzpicture}
\caption{What the bottom block of \Cref{tab:branch} measures: the same
question decided along two stacks sharing no \emph{decision} component
after the head (the named residue: emission library, endpoints) ---
the ISA-prose-derived lowering decided \emph{natively} by btormc,
against the Sail-model-derived lowering through the
\texttt{predicted}-grade bridge and Z3. Agreement corroborates the
whole diverse segment (\Cref{asm:diversity}); the stated residual
shared surface is the BTOR2 emission library and the language-owned
endpoint interpreters. The solver-level rows at the top of
\Cref{tab:branch} instead decide \emph{both} sides through the bridge
and Z3, sharing everything below the reconvergence.}
\label{fig:branch}
\end{figure}
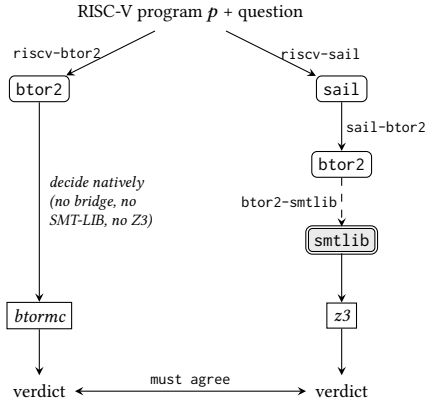

\Cref{tab:branch} (top) composes conjoined coverage along whole routes:
a construct counts iff its probe survives \emph{every} hop and every
hop with a decidable square passes it on that hop's input. Nothing a
head accepts is eaten at the hubs: every accepted probe of every route
survives composed to its route's target, squares included. The denominators are
the source language's inventory, so the two RISC-V routes are measured
on the \emph{same} 96-construct yardstick and both compose 96/96 ---
the spine composes losslessly, and the earlier version of this table,
whose two RISC-V rows quoted different pair-chosen denominators
(96 vs.\ 95), is exactly the yardstick-shopping \Cref{def:coverage}
forbids and this measurement now closes. The AArch64 rows also carry a ratchet
anecdote the machinery makes precise: one revision earlier the
Sail-mediated route composed to $0/33$, because the Sail$\to$BTOR2
translator lowered only the RV64IMC arm --- and the harness localized
every miss to the \texttt{sail-btor2} hop by name. Widening that one
hop (an additive translator bump, RISC-V arm byte-identical per
\Cref{prop:ratchet}) closed the gap without re-earning any other
evidence. A gap the system \emph{reports and localizes} is the designed
behavior; a gap papered over would be the defect.

The bottom of \Cref{tab:branch} is the fidelity payoff of
\Cref{lem:agree}: the same reachability questions decided along
independently derived routes. Twelve solver-level questions ---
constant dataflow, summation and countdown loops under unrolling,
store/load through memory, two questions about a \emph{gcc-compiled C
program} (can \texttt{a0} be 47 at halt? can it be 99?), and four
AArch64 questions decided along the Arm-manual-derived and the
Arm-Sail-model-derived route --- were each decided along both routes of
their branch; all twelve pairs of verdicts agree, including the
universal (\textsc{unreachable}) halves where agreement carries real
information. The C rows exercise the full re-establishment story: an
opaque $\trepr$ compiler heads both routes, and its output is validated
downstream twice, independently. The AArch64 trace-level rows
cross-check the same branch below the solver, directly comparing the
two carried-back behaviors under $\proj$. The table's bottom block
addresses \Cref{asm:diversity}'s scoping head-on: the twelve
solver-level rows share the BTOR2$\to$SMT-LIB bridge below the
reconvergence, so every one of the twelve is re-decided with fully
disjoint \emph{decision} stacks after the head (\Cref{fig:branch}) --- the direct
route's BTOR2 system
decided \emph{natively} by btormc (no bridge, no SMT-LIB, no Z3)
against the Sail-model route through the bridge and Z3. All twelve
agree. Every unreachable verdict among them is additionally
corroborated below the solvers: the strict BTOR2 interpreter replays
the direct route's system for the full bound and no \texttt{bad} fires
--- the tested surrogate for \Cref{thm:universal}'s
artifact-to-semantics hypothesis (recorded per question, as is the
same replay for every bounded-unreachable verdict of
\Cref{tab:bench}). The diverse segment now spans the lowering derivation
(ISA-prose vs.\ formal-model) \emph{and} the decision procedure
(different solver, different input format); the residual shared
surface, stated rather than hidden, is the BTOR2 emission library ---
the site of incident I3, the platform's one recorded cross-route
common-mode failure \emph{through shared code} (shared misreadings are
the other kind, \S\ref{sec:bugs}) --- and the language-owned endpoint
interpreters.

\subsection{The compliance slice as a question set}
\label{sec:eval-bench}

{\setlength{\tabcolsep}{4pt}
\begin{table}
\caption{The compliance slice as an external-format
question set: self-checking programs in the riscv-tests HTIF
convention (built with the pinned toolchain at the upstream link
base), each reference-run to derive per-register reachable /
bounded-unreachable questions with machine-derived ground truth,
each question decided independently along BOTH RISC-V$\to$SMT-LIB
routes. Agree = the two routes' verdicts coincide; correct = the
agreed verdict matches the interpreter-derived ground truth. Time
is both routes, all questions, end to end.}
\label{tab:bench}
\footnotesize
\begin{tabular}{@{}lrrrrrr@{}}
\toprule
Program & Steps & $k$ & Questions & Agree & Correct & Time (s) \\
\midrule
\texttt{rv64uc-compress} & 34 & 42 & 8 & 8/8 & 8/8 & 2.7 \\
\texttt{rv64ui-arith} & 63 & 71 & 8 & 8/8 & 8/8 & 8.2 \\
\texttt{rv64ui-branch} & 34 & 42 & 6 & 6/6 & 6/6 & 2.4 \\
\texttt{rv64ui-jump} & 18 & 26 & 8 & 8/8 & 8/8 & 1.4 \\
\texttt{rv64ui-ldst} & 68 & 76 & 8 & 8/8 & 8/8 & 12.0 \\
\texttt{rv64ui-logic} & 49 & 57 & 8 & 8/8 & 8/8 & 6.5 \\
\texttt{rv64ui-shift} & 45 & 53 & 8 & 8/8 & 8/8 & 5.7 \\
\texttt{rv64ui-word} & 43 & 51 & 8 & 8/8 & 8/8 & 5.0 \\
\texttt{rv64um-div} & 69 & 77 & 8 & 8/8 & 8/8 & 11.4 \\
\texttt{rv64um-mul} & 39 & 47 & 8 & 8/8 & 8/8 & 3.9 \\
\midrule
Total & & & 78 & 78/78 & 78/78 & \\
\bottomrule
\end{tabular}
\end{table}
}

The questions above were authored; \Cref{tab:bench} scales the same
discipline over a \emph{derived} question set nobody hand-picked. The
platform's RISC-V compliance slice --- self-checking programs in the
upstream riscv-tests HTIF \texttt{tohost} convention, built with the
pinned cross-toolchain at the upstream link base \texttt{0x8000\_0000}
--- doubles as a reachability benchmark: each program is reference-run
to completion, and for each data register the harness derives one
question whose target value the register demonstrably held (ground
truth \textsc{reachable}) and one whose target no register held at any
step (ground truth \textsc{unreachable} within the bound), with the
unrolling bound covering the whole run. Every question is then decided
independently along \emph{both} RISC-V routes. Three honesty notes.
First, the slice is curated in the upstream convention rather than
taken binary-for-binary from upstream: the stock \texttt{-p-} binaries
open with machine-mode CSR and trap setup that is outside the declared
user-ISA scope, so the slice re-creates the grading convention over
the user subset (the gap is typed, not hidden). Second, the ground
truth comes from the shared reference interpreter --- but that
interpreter is exactly the component the adequacy campaign validates
step-for-step against the Sail-generated gold simulator \emph{on these
same programs}, so the benchmark's ground truth is anchored outside
the platform. Third, the benchmark's discriminating power is scoped:
ground truth and both routes share the reference interpreter, so it
measures translator- and route-level defects --- a shared-misreading
defect of the MUL/ADD class (\S\ref{sec:bugs}) would pass it
undetected, and the external Sail differential, not this benchmark,
is what polices that class. The result: 78 questions over 10 programs (constant
dataflow, logic, shifts, 32-bit word ops, taken and untaken branches,
computed jumps, the store/load family, multiplication, division with
its RISC-V-defined edge cases, and compressed forms), both routes
agreeing on all 78 and all 78 matching ground truth, in under a minute
end to end --- with every one of the 39 bounded-unreachable verdicts
additionally corroborated by a full-bound replay of the direct route's
system through the strict BTOR2 interpreter (\S\ref{sec:eval-branch}'s
surrogate, recorded per question). Preparing this benchmark also forced one more fidelity
repair the base-0 probes could not see: the Sail-route front rebased
images to address 0, which silently breaks every pc-relative
absolute-address computation (\texttt{AUIPC}/\texttt{LA}) on images
linked at the upstream base --- found while squaring the translator
against the slice's link discipline, and fixed by carrying absolute
addresses through the Sail object (a versioned bump, with the squares
and this benchmark as regressions).

\subsection{End-to-end case studies}
\label{sec:eval-cases}

{\setlength{\tabcolsep}{4.5pt}
\begin{table}
\caption{End-to-end case studies. ``Replay''
is the source-level witness check of Theorem~\ref{thm:existential}:
the carried-back witness re-executed by the source interpreter.}
\label{tab:cases}
\footnotesize
\begin{tabular}{@{}>{\raggedright\arraybackslash}p{0.38\linewidth}>{\raggedright\arraybackslash}p{0.30\linewidth}cr@{}}
\toprule
Case & Verdicts & Replay & Time (s) \\
\midrule
C spine (both routes + source replay) & REACHABLE (both agree) & \checkmark & 0.21 \\
Python assert violable? (QF\_LIA) & REACHABLE & \checkmark & 0.18 \\
EVM 6*7==42 native vs bridged & REACHABLE (both agree) & \checkmark & 0.05 \\
CRN A->B twice reaches B=2 & REACHABLE & \checkmark & 3.78 \\
\bottomrule
\end{tabular}
\end{table}
}

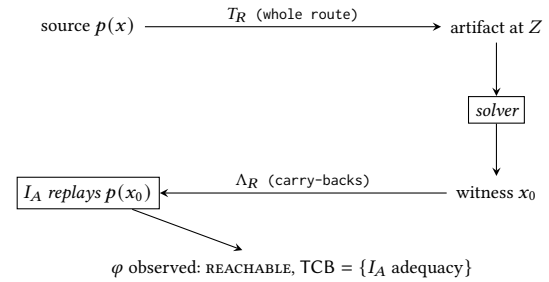
\begin{figure}
\centering
\begin{tikzpicture}[x=1cm, y=1cm]
\node[art]  (p)      at (0.9, 0)    {source $p(x)$};
\node[art]  (z)      at (6.3, 0)    {artifact at $Z$};
\node[tool] (solver) at (6.3, -1.1) {solver};
\node[art]  (model)  at (6.3, -2.2) {witness $x_0$};
\node[tool] (ia)     at (0.9, -2.2) {$I_A$ replays $p(x_0)$};
\node[art]  (v)      at (3.6, -3.2)
  {$\varphi$ observed: \textsc{reachable}, $\tcb = \{I_A$ adequacy$\}$};
\draw[chk] (p) -- node[elbl, above] {$T_R$ (whole route)} (z);
\draw[chk] (z) -- (solver);
\draw[chk] (solver) -- (model);
\draw[chk] (model) -- node[elbl, above] {$\carry_R$ (carry-backs)} (ia);
\draw[chk] (ia) -- (v);
\end{tikzpicture}
\caption{The round trip every case study in \Cref{tab:cases} completes
(\Cref{thm:existential}): the witness comes back through the route's
carry-backs and is replayed by the source interpreter. Everything on
the way down is a \emph{discovery} device --- a defect in any
translator, hub, or the solver can lose the witness, never fake the
replay.}
\label{fig:replay}
\end{figure}

\Cref{tab:cases} instantiates \Cref{thm:existential} once per hub and
front-end family, completing the round trip of \Cref{fig:replay} at
the source each time. (1)~The \emph{C spine}: both routes answer
\textsc{reachable} for \texttt{a0}${=}47$, and the source-level replay
runs the compiled program in the shared RISC-V interpreter to the halt
in \texttt{\_start()} with \texttt{a0}${=}47$ observed --- the verdict
survives even if every translator and the solver were wrong.
(2)~\emph{Python}: ``is the trailing assert violable?'' for a function
with a bounded loop and a branch; Z3 finds the witness $x{=}4$, the
shared \textsc{QF\_LIA} evaluator independently re-checks the model
(\texttt{smt\_model\_ok}), and the input is replayed through the pinned
CPython down the taken branches, firing the assert. (3)~\emph{EVM}:
$6\times7{=}42$ on the operand stack, decided \emph{natively} (btormc
on the BTOR2 system, witness replayed by the strict BTOR2 interpreter)
and \emph{bridged} (through the SMT-LIB hub with Z3) --- solve-step
corroboration across two engine families, agreeing. (4)~\emph{CRN}: a
two-firing reachability question in a reaction network, witness
schedule replayed by the Petri-net interpreter. In all four, the final
trust rests on the source interpreter (plus the small fixed replay
harness \Cref{thm:existential} names), per the last row of the
ledger in \S\ref{sec:tcb}.

\subsection{The certified tier}
\label{sec:eval-proved}

{\setlength{\tabcolsep}{4.5pt}
\begin{table}
\caption{The certified (\tprov) row inhabited,
twice: input-driven unreachability claims from a real pair,
corroborated by three engines, certified by a bit-blasted DRAT
proof, and re-validated by the independent verified checker. The
first exhibit is deliberately small (its refutation is essentially
unit propagation); the second is certificate checking at scale ---
refuting a $16{\times}16$-bit multiplication. The controls are
bogus checks that must fail --- and do.}
\label{tab:proved}
\footnotesize
\begin{tabular}{@{}p{0.30\linewidth}p{0.30\linewidth}p{0.30\linewidth}@{}}
\toprule
 & \textbf{propagation-scale} & \textbf{search-scale} \\
Question & $\mathit{helper}()^2 = 3$ (no square mod $2^{64}$) & $x \cdot y = 2^{31}{-}1$, $x,y \in [2, 2^{16}{+}1]$ (no bounded factorization of the Mersenne prime) \\
Engines agreeing \textsc{unsat} & bitwuzla, boolector, z3 & bitwuzla, boolector, z3 \\
Certificate (bitwuzla$\to$CNF, cadical$\to$DRAT) & 41746\,B DRAT, 18\,B LRAT (drat-trim, untrusted) & 2.8\,MB DRAT, 11.8\,MB LRAT \\
Independent check & cake\_lpr (\emph{formally verified}): \textbf{verified}; tier \texttt{proved} & cake\_lpr (\emph{formally verified}): \textbf{verified}; tier \texttt{proved} \\
Resulting \tcb\ (solve step) & \multicolumn{2}{l@{}}{bitwuzla:bit-blast, cake\_lpr:verified} \\
Reachable sibling & $x^2{=}4$: REACHABLE, no certificate & $x{\cdot}y{=}46341^2$: REACHABLE, no certificate \\
Negative controls & \multicolumn{2}{l@{}}{both rejected (bogus proofs vs.\ a satisfiable CNF)} \\
Wall time & 0.36\,s & 4.85\,s \\
\bottomrule
\end{tabular}
\end{table}
}

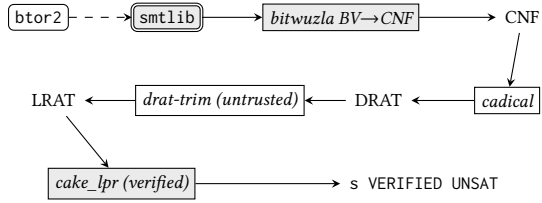
\begin{figure}
\centering
\begin{tikzpicture}[x=1cm, y=1cm]
\node[lang] (b)    at (0.5, 0)    {btor2};
\node[hub]  (smt)  at (2.2, 0)    {smtlib};
\node[tcb]  (bb)   at (4.5, 0)    {bitwuzla BV$\to$CNF};
\node[art]  (cnf)  at (6.9, 0)    {CNF};
\node[tool] (cad)  at (6.7, -1.1) {cadical};
\node[art]  (drat) at (5.0, -1.1) {DRAT};
\node[tool] (dt)   at (2.9, -1.1) {drat-trim (untrusted)};
\node[art]  (lrat) at (0.7, -1.1) {LRAT};
\node[tcb]  (clpr) at (1.6, -2.2) {cake\_lpr (verified)};
\node[art]  (ok)   at (5.6, -2.2) {\texttt{s VERIFIED UNSAT}};
\draw[pred] (b) -- (smt);
\draw[chk] (smt) -- (bb);
\draw[chk] (bb) -- (cnf);
\draw[chk] (cnf) -- (cad);
\draw[chk] (cad) -- (drat);
\draw[chk] (drat) -- (dt);
\draw[chk] (dt) -- (lrat);
\draw[chk] (lrat) -- (clpr);
\draw[chk] (clpr) -- (ok);
\end{tikzpicture}
\caption{The pipeline \Cref{tab:proved} evaluates: the query is
bit-blasted to CNF, cadical's DRAT refutation is elaborated by
\texttt{drat-trim} --- an untrusted step that can only fail the
re-check, never fake a pass --- and \texttt{cake\_lpr}, the formally
verified checker, re-validates the LRAT against the CNF from scratch.
Shaded is what stays in the verdict's \tcb: the bit-blast and the
verified checker; the three agreeing engines and the elaborator are
discovery devices.}
\label{fig:cert}
\end{figure}

\Cref{tab:proved} inhabits the ledger's certified row
(\S\ref{sec:tcb}), twice --- via the pipeline of \Cref{fig:cert} --- once at each end of the difficulty
scale. Both questions are pair-derived and input-driven: eBPF programs
whose \texttt{CALL} models helper returns as free inputs. The first
squares one input and asks whether $x^2 = 3$; no square exists modulo
$2^{64}$, so the claim is universal over all inputs --- and its
refutation is essentially unit propagation on the low bits (the
elaborated LRAT is 18 bytes). The second masks and offsets two inputs
into $[2, 2^{16}{+}1]$, multiplies them, and asks whether the product
can be $2^{31}-1$: the Mersenne prime admits no factorization with
both factors in range, and refuting a $16{\times}16$-bit
multiplication is certificate checking at scale --- a 2.8\,MB DRAT
elaborating to an 11.8\,MB LRAT, re-validated in seconds. In both,
three engines corroborate \textsc{unsat}, bitwuzla bit-blasts the
query to CNF, cadical emits a DRAT refutation, \texttt{drat-trim}
elaborates it to LRAT --- an \emph{untrusted} step: a wrong
elaboration can only fail the re-check, never fake a pass --- and
\texttt{cake\_lpr}, the \emph{formally verified} CakeML LRAT checker
(soundness machine-proved down to the binary; here built natively
from its compiler-generated ARMv8 assembly), re-validates the proof
against the CNF from scratch.
The verdicts' \tcb\ is exactly the ledger's certified row: the
bit-blaster's BV$\to$CNF step and a verified checker, with all three
deciding engines and the elaborator demoted to discovery devices ---
the solve step's residue; the route hops that produced the artifact
remain accounted per the ledger. Both certified systems are also
corroborated below the whole certificate pipeline: fifty seeded random
helper streams replayed through the strict BTOR2 interpreter fire no
\texttt{bad} (recorded in the run data). The
reachable siblings ($x^2 = 4$; $x{\cdot}y = 46341^2$) correctly
decline to certify.

Wiring this exhibit produced one more architecture-caught defect,
live: the checker \emph{adapter} matched the substring
\texttt{VERIFIED} in the checker's output --- and \texttt{drat-trim}
reports failure as \texttt{s NOT VERIFIED}, which contains it. Every
check, sound or bogus, reported success. The bug could not surface on
hosts without the checker (the call raised before parsing) nor on real
certificates (which do verify); it fell to the \emph{negative
control} --- a bogus refutation of a satisfiable CNF that must fail
and did not. The fix parses the exact status line, and the control is
now a permanent test. This is the paper's own instrument audited by
the discipline it describes (cf.\ \S\ref{sec:bugs}), and a caution we
commend to anyone wiring a proof checker: \emph{a checker adapter
without a negative control is itself unchecked}. The verified
checker's adapter was written under that rule from the start --- wisely,
since \texttt{cake\_lpr} exits 0 even when checking fails; only the
exact \texttt{s VERIFIED UNSAT} line signals success, and both
checkers' negative controls are permanent tests. One honest subtlety
the exercise surfaced: for this instance the refutation is
propagation-dominated, so mutating single proof lines is not a valid
negative control (the checker legitimately completes the derivation);
the satisfiable-formula control is the sound one, since no valid
refutation of a satisfiable formula exists.

\subsection{Defects the architecture caught}
\label{sec:bugs}


The pairs were written by independent LLM agents with no human semantic
review (\S\ref{sec:agents}); the architecture's checks were the only
gate. Mining the full history (675 commits across all refs at the
2026-07-03 mining pass; $\sim$21
runtime-code fix commits) plus the three defects the conjoined-coverage
measurement caught on its first run (\S\ref{sec:eval-capability}), the
probe-operand deficiency the fault-injection experiment exposed, and
the compare-vector gap only the common-mode round could expose
(\S\ref{sec:eval-escape})
yields 24 recorded defect incidents: 15
caught by the architecture's cross-checks rather than by ordinary unit
tests, and 7 more --- including the checker-adapter defect of
\S\ref{sec:eval-proved}, a probe that could never have exercised its
own square, probes whose operands underdetermined their constructs,
and benchmark compare vectors on which signedness is invisible
--- by the audits and controls the architecture
imposes on its own instruments, and the remaining two otherwise (one
by machine-vs-human label disagreement, one only partially
attributed). Read this as case-study evidence, not a
hit rate: the catalog is mined from our own history, it has no
denominator (the blind-spot incident below proves the escape rate is
nonzero), and no baseline says what a conventional test regime would
have caught. \Cref{tab:bugs} lists ten
representative architecture-caught defects, each localized ---
component, step, observable --- as \Cref{cor:localization} promises.

\begin{table}
\caption{Defects caught by the architecture (selection; full list with
commit-level evidence in the artifact). ``Caught by'' names the check
that fired --- the cross-check layers of \S\ref{sec:tcb}, plus the
platform's own audits and controls.}
\label{tab:bugs}
\footnotesize
\begin{tabular}{@{}>{\raggedright\arraybackslash}p{0.42\linewidth}>{\raggedright\arraybackslash}p{0.17\linewidth}>{\raggedright\arraybackslash}p{0.31\linewidth}@{}}
\toprule
\textbf{Defect} & \textbf{Component} & \textbf{Caught by} \\
\midrule
bv64 sort hardcoded for \texttt{and/or/xor} over bv1 operands ---
malformed BTOR2 that real solvers tolerated &
RISC-V $\to$ BTOR2 translator &
strict in-process BTOR2 evaluator during the square's alignment walk \\
unsigned loads (\texttt{LBU/LHU/LWU}) missing zero-extension to bv64 &
RISC-V $\to$ BTOR2 translator &
solver leg: Z3 sort rejection on the corpus \\
\texttt{init} value nid emitted after the state nid --- rejected by every
conformant tool, tolerated by the one wired solver; affected every
stateful pair &
shared BTOR2 emitter &
solver corroboration: wiring the 2nd/3rd engines (btormc, pono) made the
latent malformation manifest \\
\texttt{local.get} wrote a bv32 node into the bv64 stack array &
Wasm $\to$ BTOR2 translator &
solver leg: Z3 sort mismatch at BMC \\
carry-back silently zero-filled registers the solver witness omits
(init-pinned states) --- every pinned task's entry state misread &
witness carry-back $\carry$ &
witness-replay anchor audit (mismatch chain) \\
default BMC bound mapped ``no violation within $k$'' to
\textsc{unreachable} --- a false negative at step 93 &
verdict semantics &
oracle vs.\ corpus label \\
ELF header bytes decoded as instructions --- interpreter tolerated it,
the square did not &
shared ELF loader &
square oracle on the first real gcc binary \\
BTOR2 \texttt{constraint} lines silently dropped in the SMT-LIB
unrolling --- an under-constrained, soundness-leaking encoding &
BTOR2 $\to$ SMT-LIB bridge &
coverage probe: driving the bridge to its 56/56 spec inventory exposed
the hole \\
\texttt{JUMP/JUMPI} underflow-halt row diverged on \texttt{pc} &
EVM $\to$ BTOR2 translator &
square step alignment during widening \\
\texttt{INT\_MIN / -1}: C-undefined vs.\ RV64-defined divergence, plus a
CBMC false positive on the same task &
(real route divergence) &
C-vs-ISA differential branch, with a UB-vs-fault classifier
adjudicating \\
\bottomrule
\end{tabular}
\end{table}

Three second-order observations are worth as much as the table. First,
\emph{the checks caught their own instruments} --- seven of the 24
incidents are check-of-the-check repairs; three representatives: the Sail
ISA differential was once passing vacuously (comparing empty traces),
the in-image suite caught a witness generator emitting empty witnesses,
and the DRAT checker adapter accepted every outcome until a negative
control ran (\S\ref{sec:eval-proved}) --- all found because the
harnesses carry positive \emph{and negative} controls, and all three
are why \Cref{asm:adequacy} is stated as an assumption with empirical
discharge rather than as a fact. Second, the history contains a genuine
\emph{blind spot} of the square, exactly of the class \Cref{lem:agree}
predicts: the RISC-V interpreter \emph{and} the RISC-V$\to$BTOR2
translator both silently mis-decoded \texttt{MUL} (\texttt{funct7}
\texttt{0x01}) as
\texttt{ADD} --- both legs agreed, the square was blind, and the defect
was found by manual audit. The platform's structural answer is
diversity: the decoder is now anchored against 610 Sail-emitted words,
and unknown \texttt{funct7} hard-aborts on both sides. Third, in one incident
three existing unit tests had been written \emph{against the buggy
behavior} (an evaluator masking every array store to 8 bits): tests
written by a component's author codify the author's misreading;
cross-checks against independently derived semantics do not. The
architecture also \emph{disconfirmed} one reported bug (a claimed
slice-width defect shown to be a stale-cache artifact), and its
differential campaigns --- 300 seeded RV64IMC programs vs.\ the Sail
simulator, a 10-seed Csmith campaign vs.\ native gcc, 463 reference
cases --- completed with zero divergences against harnesses whose
positive controls confirm they can fail.

\paragraph{Two incidents in full}
The catalog's two most argument-bearing entries deserve their
narratives. First, the \emph{blind spot}: the RISC-V interpreter
\emph{and} the \texttt{riscv-btor2} translator mis-decoded
\texttt{MUL} (\texttt{funct7} \texttt{0x01}) as \texttt{ADD} ---
identically. Both legs of every square agreed, so the square was
structurally blind, and the defect was found by manual audit during
widening, not by the architecture. This is the residual failure mode
the calculus itself predicts --- \Cref{lem:agree}'s ``both unfaithful
with identical projected error'' --- landing not between two
translators but between a translator and the component
\Cref{thm:existential} crowns as the sole residual \tcb. Both sides
now hard-abort on unknown \texttt{funct7}, and the structural answer
is \emph{external} anchoring: the decoder is validated against 610
Sail-emitted instruction words and against an independently written
RV64C decompressor, construct by construct --- and the common-mode
round of \S\ref{sec:eval-escape} now measures this incident's whole
class directly. Second, the \emph{common
mode}: the shared BTOR2 emission library declared states before their
constant \texttt{init} values, so every emitted \texttt{init} line
was malformed --- in every BTOR2-emitting pair at once, both branches
of every RISC-V question included --- and the one solver then wired
(Z3, through the bridge) tolerated the malformation. Adding the
second and third engines turned it into hard rejections: solve-step
corroboration, not the squares, caught it (incident I3). Between them
the two incidents mark the exact boundary of \Cref{asm:diversity}:
shared code and shared misreadings are what agreement cannot see,
which is why the ledger keeps the emission library and the endpoint
interpreters inside every branch verdict's \tcb\ --- and why the
disjoint stacks of \Cref{fig:branch} and the external differentials
exist.

\subsection{An escape-rate estimate for the gate}
\label{sec:eval-escape}

\begin{table}
\caption{Escape-rate experiment: seeded semantic
mutations of the riscv-btor2 emissions (uniform rules model
systematic mis-lowerings, site rules single-site defects;
operand swaps only on non-commutative operators; rules that
change no probe artifact are excluded from the denominator),
each run through the architecture's gates in order. ``Square''
is the conjoined probe suite (\S\ref{sec:eval-capability}),
``branch'' the authored solver questions against the intact Sail
route, ``bench'' the compliance-derived ground-truth questions
(\S\ref{sec:eval-bench}). Zero escapes means the seeded mutation
families --- catalog-derived, with probes hardened against this
experiment's own earlier rounds --- are now covered; it is not an
estimate that the escape rate is zero (\S\ref{sec:eval-escape}).}
\label{tab:escape}
\footnotesize
\begin{tabular}{@{}lr@{}}
\toprule
Applicable mutants & 55 \\
\midrule
Caught by the square suite & 51 \\
Caught by branch agreement & 0 \\
Caught by the derived benchmark & 4 \\
Escaped all three gates & \textbf{0} \\
\bottomrule
\end{tabular}
\end{table}

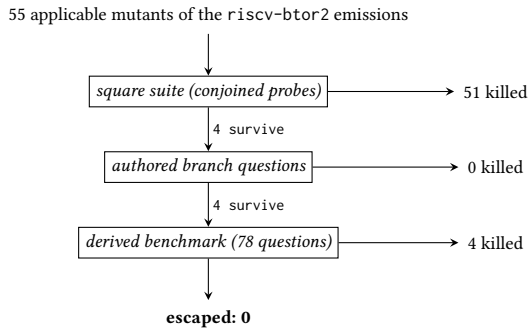
\begin{figure}
\centering
\begin{tikzpicture}[x=1cm, y=1cm]
\node[art]  (m)  at (2.4, 0)    {55 applicable mutants of the
  \texttt{riscv-btor2} emissions};
\node[tool] (g1) at (2.4, -1.0) {square suite (conjoined probes)};
\node[tool] (g2) at (2.4, -2.0) {authored branch questions};
\node[tool] (g3) at (2.4, -3.0) {derived benchmark (78 questions)};
\node[art]  (e)  at (2.4, -4.0) {\textbf{escaped: 0}};
\node[art] (k1) at (6.2, -1.0) {51 killed};
\node[art] (k2) at (6.2, -2.0) {0 killed};
\node[art] (k3) at (6.2, -3.0) {4 killed};
\draw[chk] (m) -- (g1);
\draw[chk] (g1) -- node[elbl, right] {4 survive} (g2);
\draw[chk] (g2) -- node[elbl, right] {4 survive} (g3);
\draw[chk] (g3) -- (e);
\draw[chk] (g1) -- (k1);
\draw[chk] (g2) -- (k2);
\draw[chk] (g3) -- (k3);
\end{tikzpicture}
\caption{The escape-rate experiment (\Cref{tab:escape}) as a gate
stack: every applicable mutant runs the gates in the order the
platform applies them, and each gate's kill count is measured. The
authored branch questions catch nothing the square suite missed ---
the derived benchmark is the differentiator (a finding in itself).}
\label{fig:escape}
\end{figure}

The catalog above is mined from history and has no denominator; the
MUL/ADD blind spot proves the gate's escape rate is nonzero.
\Cref{tab:escape} measures it (\Cref{fig:escape} draws the gate
stack): seeded semantic mutations of the
\texttt{riscv-btor2} translator's emissions --- uniform rules modeling
systematic mis-lowerings (every \texttt{sext} emitted as \texttt{uext},
incident I2's family) and site rules modeling single-site defects
(incident I1's family); operand swaps only on non-commutative
operators, so no trivially-equivalent mutants; rules that change no
probe artifact excluded from the denominator --- each run through the
architecture's gates in the order the platform applies them. The
experiment earned its keep before its final numbers did: in its first
round, both \texttt{srl}$\to$\texttt{sra} mutants escaped \emph{all
three} gates (36 of 51 caught at the square layer), because every
ALU probe ran its construct on degenerate all-zero operands, on which
logical and arithmetic shifts agree --- the ``96/96 conjoined'' verdict
was faithfulness \emph{on the probes}, and the probes underdetermined
the constructs. Hardening the probes with mixed-sign operands moved 50
of 53 catches to the square layer but let \texttt{ult}$\to$\texttt{ulte}
escape --- strictness is observable only at \emph{equal} operands ---
so the probes now run an equal-operand and a mixed-sign instance per
construct (incident I23; the ratchet kept every previously covered
construct covered throughout). In the final round, 51 of 55 applicable
mutants die at the square suite, the remaining 4 at the
compliance-derived benchmark, and none escape. Three honest notes:
zero escapes means this \emph{mutation family} is now covered, not
that the escape rate is zero;
the authored branch questions caught nothing in any round (the derived
benchmark subsumes them as a differentiator, a finding in itself); and
the denominator is one translator's emission space, not the platform.

\paragraph{The common-mode round}

\begin{table}
\caption{Common-mode (both-leg) fault injection:
the same misreading injected into BOTH the reference interpreter
and the riscv-btor2 translator --- the MUL/ADD class of
\S\ref{sec:bugs}, which single-leg mutation cannot model. The
square suite is structurally blind on every mutant (both legs
wrong identically), and so is expected-based grading in the
poisoned world, where the benchmark's ground truth is derived by
the mutated interpreter itself. What catches every mutant are the
gates \emph{outside} the shared misreading: the independently
derived Sail route disagreeing (cross-route), with the external
Sail-simulator differential --- run as a parallel column --- also
catching every mutant.}
\label{tab:common}
\footnotesize
\setlength{\tabcolsep}{2pt}
\begin{tabular}{@{}lllll@{}}
\toprule
Misreading (both legs) & Square & Poisoned & Caught by & Diff.\ \\
\midrule
MUL as ADD & blind & blind & cross-route (rv64um-mul) & \checkmark \\
SUB as ADD & blind & blind & cross-route (rv64ui-arith) & \checkmark \\
SRA as SRL & blind & blind & cross-route (rv64ui-shift) & \checkmark \\
SLT as SLTU & blind & blind & Sail differential & \checkmark \\
AND as OR & blind & blind & Sail differential & \checkmark \\
XOR as AND & blind & blind & cross-route (rv64ui-logic) & \checkmark \\
\bottomrule
\end{tabular}
\end{table}

Single-leg mutation cannot model the MUL/ADD class where both legs of
the square are wrong \emph{identically} --- so a second experiment
injects exactly that (\Cref{tab:common}): six shared misreadings,
each applied as the same uniqueness-checked source patch to \emph{both}
the reference interpreter and the \texttt{riscv-btor2} translator
(in-memory shadows; the repository is untouched). The structural
predictions hold, measured: the square suite is blind on all six (both
legs agree on every probe), and so is expected-based grading in the
\emph{poisoned world}, where the benchmark's ground truth is derived
by the mutated interpreter itself. What catches the mutants are the
gates \emph{outside} the shared misreading ---
\Cref{fig:escape}'s stack read as ordered by independence, each gate
derived farther from the shared code than the last: four die at
cross-route
disagreement with the Sail-derived branch, and two survive even that
--- their misreadings never surface in the per-register values the
derived questions sample --- falling only to the external
Sail-simulator differential, which (run as a parallel column) catches
all six. The round also caught its own instrument, live: in its first
run, \texttt{slt-as-sltu} \emph{escaped every gate} --- the slice's
compare vectors were all same-sign, on which signed and unsigned
comparison agree, so neither the derived questions nor the
differential's executed streams could see the difference (incident
I24, the slice-level instance of I23's probe lesson, and the
platform's first measured full-stack escape). The ratchet's answer
was upstream-faithful --- \texttt{rv64ui-slt}'s mixed-sign and
equal-operand vectors added to the slice --- after which the round
completes with zero escapes. The honest reading mirrors
\Cref{asm:diversity}'s scope exactly: against shared misreadings every
gate built from the shared reading is blind \emph{by construction},
corroboration is only as
strong as the most independently derived anchor, and that
anchor's vectors are themselves an instrument the discipline must
audit.

\subsection{The other direction: an LLM plays the platform}
\label{sec:eval-player}

\begin{table*}
\caption{The LLM-player experiment: 12
ground-truthed questions, two arms. Arm A (unaided) answers from
reasoning alone; arm B answers \emph{via the platform}. A
\checkmark{} means the verdict matches the platform-established
ground truth (R = reachable, U = unreachable within the stated
bound). Both arms are correct on all 12 --- the contrast is the
final column: every arm-B verdict carries a machine-checked
evidence artifact, where arm A rests on the model's say-so (all
arm-A answers were reported at high confidence). Full question
texts, per-run transcripts, and verbatim evidence bases:
\texttt{results/llm\_player/} in the artifact.}
\label{tab:player}
\footnotesize
\begin{tabular}{@{}llcccl@{}}
\toprule
Q & Question & Truth & Arm A & Arm B & Arm-B evidence artifact \\
\midrule
R1 & RISC-V loop: $x_1{=}15$ within $k{=}25$ & R & \checkmark & \checkmark & two-route agreement \\
R2 & RISC-V loop: $x_1{=}16$ within $k{=}25$ & U & \checkmark & \checkmark & two-route agreement (bounded) \\
R3 & \texttt{srli} (logical) $\mathtt{0xF..F8} \gg 60$: $x_1{=}15$ & R & \checkmark & \checkmark & two-route agreement \\
R4 & \texttt{srai} (arithmetic) $\gg 60$: $x_1{=}15$ & U & \checkmark & \checkmark & two-route agreement (bounded) \\
R5 & \texttt{lb} (sign-ext.) of byte \texttt{0x84}: $x_3{=}$\texttt{0x84} & U & \checkmark & \checkmark & two-route + self-devised positive control \\
R6 & \texttt{lbu} (zero-ext.) of byte \texttt{0x84}: $x_3{=}$\texttt{0x84} & R & \checkmark & \checkmark & two-route agreement \\
E1 & eBPF $x{\cdot}x \bmod 2^{64} = 3$ & U & \checkmark & \checkmark & \tprov: LRAT re-validated by cake\_lpr \\
E2 & eBPF $x{\cdot}x = 4$ & R & \checkmark & \checkmark & z3 + native btormc; witness replayed \\
E3 & $x{\cdot}y = 1073741789$, $x,y \in [2, 2^{16}{+}1]$ & U & \checkmark & \checkmark & \tprov: 17\,MB LRAT via cake\_lpr \\
E4 & $x{\cdot}y = 2147766287$, same range & R & \checkmark & \checkmark & witness replay; factor pair exhibited \\
P1 & Python assert: $y{=}16$ violable & R & \checkmark & \checkmark & SMT model re-checked + CPython replay \\
P2 & Python assert: $y{=}15$ violable & U & \checkmark & \checkmark & \textsc{unsat}; per-run-checked \tpred\ route \\
\bottomrule
\end{tabular}
\end{table*}

\begin{figure}
\centering
\begin{tikzpicture}[x=1cm, y=1cm]
\node[art] (q) at (3.6, 0)
  {reachability question + platform-established ground truth};
\node[tool, align=center] (aa) at (1.4, -1.2)
  {arm A: unaided LLM\\(question text only, no tools)};
\node[tool, align=center] (ab) at (5.9, -1.2)
  {arm B: LLM player\\(same question + repository)};
\node[hub, align=center] (plat) at (5.7, -2.4)
  {\hg: routes, squares,\\solvers, checkers};
\node[art, align=center] (va) at (1.4, -3.5)
  {verdict\\(say-so, high confidence)};
\node[art, align=center] (vb) at (5.9, -3.5)
  {verdict + machine-checked\\evidence artifact};
\node[art] (g) at (3.6, -4.5) {graded against ground truth};
\draw[chk] (q) -- (aa);
\draw[chk] (q) -- (ab);
\draw[chk] (ab) -- (plat);
\draw[chk] (plat) -- (vb);
\draw[chk] (aa) -- (va);
\draw[chk] (va) -- (g);
\draw[chk] (vb) -- (g);
\end{tikzpicture}
\caption{The two-arm player protocol of \Cref{tab:player}: one fresh,
context-free agent per question per arm, no retries. Both arms are
graded against the same platform-established ground truth; what
differs by construction is what the verdict \emph{rests on} --- arm B
can take no unchecked step, so every answer carries the evidence
artifact of \Cref{tab:player}'s final column.}
\label{fig:player}
\end{figure}
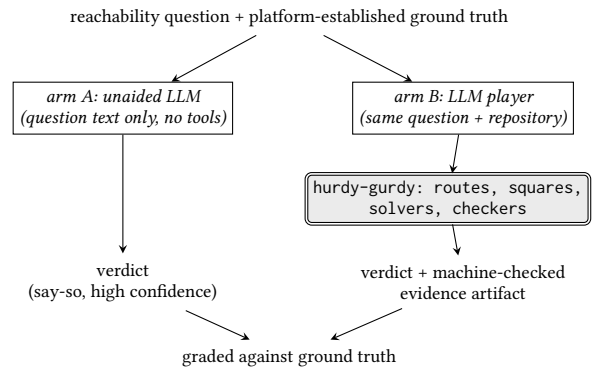

The introduction's two-directional experiment has, until here, reported
only its build half. This subsection is a first, deliberately small
evaluation of the \emph{player} half: 12 reachability questions with
platform-established ground truth (6 reachable / 6 unreachable, spanning
RISC-V signedness and load-width traps, eBPF square-residue and
bounded-factorization questions, and Python assert-violability), posed
to fresh, context-free frontier-LLM agents under two protocols
(\Cref{fig:player}). Arm A
(unaided) received only the question text, instructed to answer by
reasoning with no tools. Arm B (player) received the same question plus
the repository and the head-construction snippet, instructed to decide
\emph{via the platform} and report the verdict with its evidence basis.
One agent per question per arm, no retries; \Cref{tab:player} lists the
questions and both arms' verdicts with each arm-B answer's evidence
artifact; prompts, ground truths,
per-run outputs, and grading are in the artifact
(\texttt{results/llm\_player/}).

Both arms scored 12/12. That is the honest headline, and it reframes
what the platform buys at this question scale: not raw accuracy ---
the unaided frontier model simulated the loops, caught both signedness
traps, and factored 2147766287 by Fermat's method in its head --- but
the \emph{evidence class} of the answer. Every arm-A verdict rests on
the model's say-so (all were reported at high confidence, so
confidence would not have flagged a wrong one); every arm-B verdict
carries a machine-checked artifact: two-route agreement on the RISC-V
questions, hedged explicitly as bounded-$k$ verdicts resting on route
fidelity; witness replays through the strict interpreters on the
reachable eBPF/Python questions; and tier-raises to \tprov\ on the
unreachable eBPF questions. The sharpest contrast is the
bounded-factorization question: arm A answered correctly by
\emph{recalling} that 1073741789 is the largest prime below $2^{30}$
--- opaque memory, exactly the evidence the calculus cannot grade ---
while arm B produced a 17\,MB LRAT certificate re-validated by the
formally verified checker, with the trusted base named. One player
even devised its own positive control unprompted (flipping the load
question's target to the sign-extended value and confirming
\textsc{reachable} on both routes). Four limitations, stated plainly:
the protocol scripted which platform entry points to use (the players
executed and interpreted; they did not discover the platform cold);
the subject model is from the same family that built the platform;
the grading is partly circular --- ground truth is
platform-established and arm B answers \emph{via} the platform, so arm
B can fail essentially only by mis-operating or misreading the tools;
and
12 questions at this size cannot separate the arms on correctness ---
questions hard enough that unaided reasoning actually fails are the
obvious next step.\ifarxiv{} (Both protocol limitations are removed,
and the separation produced, by the post-snapshot rerun of
\S\ref{sec:eval-player2}.)\fi

\subsection{Performance, determinism, and scale}
\label{sec:eval-perf}

\begin{table}
\caption{Route cost for the RISC-V loop question
($k{=}25$ unrolling): whole-route translation, cold vs.\ warm (the
content-addressed cache of Proposition~\ref{prop:cache}), Z3 decide
time, byte-determinism (twice-and-diff), and artifact size.}
\label{tab:perf}
\footnotesize
\begin{tabular}{@{}lrrrcr@{}}
\toprule
Route & cold (ms) & warm (ms) & decide (ms) & det. & artifact \\
\midrule
RISC-V, direct & 5 & $<$0.1 & 55 & \checkmark & 201\,kB \\
RISC-V, via Sail & 6 & 0.1 & 44 & \checkmark & 209\,kB \\
\bottomrule
\end{tabular}
\end{table}

{\setlength{\tabcolsep}{4pt}
\begin{table}
\caption{Scalability probes on two axes. Left:
the summation-loop question at growing bounds ($x_1 = N(N{+}1)/2$
after $N$ iterations; reachable at every size), whole-route
translation and Z3 decide time. Right: the certified tier on the
bounded-factorization family at growing factor widths ---
certificate sizes grow from tens of kilobytes to megabytes and cake\_lpr
re-validates each.}
\label{tab:scale}
\footnotesize
\begin{tabular}{@{}lllrr@{}}
\toprule
Loop question & verdict & translate (s) & decide (s) & artifact \\
\midrule
$N{=}5$, $k{=}30$ & reachable & 0.01 & 0.04 & 240\,kB \\
$N{=}20$, $k{=}105$ & reachable & 0.02 & 0.10 & 830\,kB \\
$N{=}50$, $k{=}255$ & reachable & 0.04 & 0.27 & 2054\,kB \\
$N{=}100$, $k{=}505$ & reachable & 0.07 & 0.55 & 4093\,kB \\
\bottomrule
\end{tabular}

\medskip

\begin{tabular}{@{}llllr@{}}
\toprule
Certified question & outcome & DRAT & LRAT & time (s) \\
\midrule
8-bit factors, target 65521 & certified & 27.3\,kB & 24.7\,kB & 0.2 \\
12-bit factors, target 16777213 & certified & 0.2\,MB & 64.1\,kB & 0.3 \\
16-bit factors, target 2147483647 & certified & 2.8\,MB & 11.8\,MB & 4.9 \\
\bottomrule
\end{tabular}
\end{table}
}

Translation is not the cost center\ifarxiv{} (\Cref{tab:perf})\fi: whole-route translation of the loop
question ($k{=}25$ unrolling, $\sim$200\,kB of SMT-LIB) is
$\sim$5\,ms cold and a sub-millisecond content-addressed cache hit warm
(\Cref{prop:cache} cashed in), against a Z3 decide in the tens of
milliseconds on either
route; the
twice-and-diff determinism check passes on both routes. The numbers say
the honest overhead of the whole discipline --- checking every hop,
re-running for determinism, replaying witnesses --- is small against
the solver call it wraps. \Cref{tab:scale} probes where the pipeline
stands as questions grow, on two axes: unrolling the summation loop to
$N{=}100$ ($k{=}505$, a 4\,MB artifact) keeps whole-route translation
under 0.1\,s and the decide under 1\,s, growing roughly linearly; and
the certified tier survives the bounded-factorization family from
propagation-size to search-size --- certificates growing from 25\,kB to
an 11.8\,MB LRAT, each re-validated by the verified checker in seconds.
Neither axis reaches industrial scale, but both cross the region where
``toy'' criticisms live: the pipeline's costs grow smoothly, with no
cliff through the measured range.

\ifarxiv
\subsection{Post-snapshot: constraint enforcement across engines}
\label{sec:eval-constraint}

\begin{table}
\caption{The constrained corpus: nine authored
BTOR2 systems with by-construction ground truth
(5 reachable / 4 unreachable),
each decided natively (btormc, canary-controlled bounded
exhaustion), bridged (the per-frame encoding, z3), and
corroborated by the shared evaluator --- witness replay on the
reachable rows (both the bridged model's carry-back and btormc's
\texttt{.wit}), seeded no-bad runs on the unreachable rows
(\textsc{r} = reachable, \textsc{u} = unreachable within the
bound). Below the rule, the structural controls; the first keeps
the historical global constraint reading as an instrument and
measures what it masks. Post-snapshot measurement at commit
\texttt{758b1f2} (2026-07-15).}
\label{tab:constraint}
\footnotesize
\setlength{\tabcolsep}{2pt}
\begin{tabular}{@{}lrllll@{}}
\toprule
System & $k$ & Truth & Native & Bridged & Evaluator \\
\midrule
\texttt{valid-prefix-reach} & 5 & \textsc{r} & \textsc{r} & \textsc{r} & replay \checkmark \\
\texttt{two-guards-reach} & 6 & \textsc{r} & \textsc{r} & \textsc{r} & replay \checkmark \\
\texttt{late-window-reach} & 8 & \textsc{r} & \textsc{r} & \textsc{r} & replay \checkmark \\
\texttt{vacuous-constraint} & 4 & \textsc{r} & \textsc{r} & \textsc{r} & replay \checkmark \\
\texttt{constrained-input-reach} & 2 & \textsc{r} & \textsc{r} & \textsc{r} & replay \checkmark \\
\texttt{constraint-blocked-input} & 3 & \textsc{u} & \textsc{u} & \textsc{u} & no bad \checkmark \\
\texttt{truncated-before-bad} & 12 & \textsc{u} & \textsc{u} & \textsc{u} & no bad \checkmark \\
\texttt{invalid-row-bad} & 5 & \textsc{u} & \textsc{u} & \textsc{u} & no bad \checkmark \\
\texttt{bound-scoped-unreach} & 1 & \textsc{u} & \textsc{u} & \textsc{u} & no bad \checkmark \\
\midrule
\multicolumn{3}{@{}l}{global reading, \texttt{valid-prefix-reach}} & \multicolumn{3}{l@{}}{\textsc{unreach} --- the reach masked \checkmark} \\
\multicolumn{3}{@{}l}{vacuous constraint vs.\ none} & \multicolumn{3}{l@{}}{verdicts equal \checkmark} \\
\multicolumn{3}{@{}l}{blocked input, constraint removed} & \multicolumn{3}{l@{}}{flips to \textsc{reach} \checkmark} \\
\bottomrule
\end{tabular}
\end{table}

The first of the exhibit families to graduate to a measurement ---
stamped with its own commit in \Cref{tab:constraint}'s caption; unlike
every table above, its numbers postdate the snapshot. The shared BTOR2
evaluator enforces environment \texttt{constraint}s with the per-frame
reading (per-row observables; a violating row ends the run; a
\texttt{bad} counts only on a constraint-valid row), and wiring that
enforcement surfaced a latent bridge defect of independent interest:
asserting constraints \emph{globally} over the unrolling masks a bad
reached on a valid prefix before a later, inevitable violation --- and
disagrees with the native checkers on exactly that system.
\Cref{tab:constraint} turns that episode into a benchmark: nine
authored systems whose ground truth holds by construction --- reaches
on valid prefixes (single and double guards, short and long windows,
through a constrained input), and unreaches that exist only because of
a constraint, only beyond a truncation, only on an invalid row, or
only beyond the bound --- each decided natively (btormc), bridged
(the per-frame encoding, z3), and corroborated by the shared
evaluator: witness replay on every reachable row (both the bridged
model's carry-back and btormc's \texttt{.wit} replayed through the
strict interpreter), seeded no-bad runs on every unreachable row. All
nine agree with ground truth in both polarities, and the defective
global reading, kept as an instrument, still masks the valid-prefix
reach --- the permanent negative control, beside the additive one (a
vacuous constraint moves no verdict) and the blocking one (removing
the constraint flips the blocked system back). Three honesty notes:
the corpus is authored and deliberately tiny (nine systems, one
operator family), constructed to pin the semantic corners of the
\texttt{constraint} directive rather than to sample a population; the
evaluator \emph{corroborates} --- replay and seeded runs --- it never
decides; and agreement here shares the BTOR2 emission surface all
three deciders consume, so a shared misreading of the directive itself
would pass --- the mutation discipline of \S\ref{sec:eval-escape}
applies and is future work for this family
(\texttt{tools/constraint\_corpus.py}; the experiment is also a
permanent test).

\subsection{Post-snapshot: cost calibration for the books}
\label{sec:eval-costs}

\begin{table}
\caption{Cost calibration for the books:
5 repeated capped route-grader runs over the
RISC-V head (fresh subprocess and ledger file each, so the
translation cache swallows no record; the grader's own probe
cap applies). Top: pooled per-hop cost profiles --- translate
and square-oracle (o) medians --- and the bridged decide per
engine. Middle: per-repetition route translate totals (median
and worst repetition) and, per route, in how many repetitions
the annotated report marked it dominated --- the mark's
coherence and stability are the measurement; both routes stay
listed at equal assurance and direction, no scalar ranking.
Bottom: the honesty invariants, executable.
Timings are host-tagged (single machine); the corpus is the
grader's capped slice, so the numbers calibrate the
\emph{instrument}, not the platform's workload. Post-snapshot
measurement at commit \texttt{b6e1d89}
(2026-07-15).}
\label{tab:costs}
\footnotesize
\setlength{\tabcolsep}{3.5pt}
\begin{tabular}{@{}lrrrr@{}}
\toprule
Hop / engine & $n$ & median (ms) & p90 (ms) & o (ms) \\
\midrule
\texttt{btor2-smtlib} & 235 & 0.543 & 0.733 & --- \\
\texttt{riscv-btor2} & 120 & 0.550 & 0.715 & 1.439 \\
\texttt{riscv-sail} & 120 & 0.016 & 0.025 & 0.083 \\
\texttt{sail-btor2} & 120 & 0.558 & 0.727 & 1.465 \\
decide[\texttt{z3}] & 10 & 3.305 & 6.870 & --- \\
\midrule
Route (translate total) & reps & median & max & marked \\
\midrule
direct (\texttt{riscv-btor2}) & 5 & 1.084 & 1.109 & 1/5 \\
via Sail (3 hops) & 5 & 1.102 & 1.144 & 4/5 \\
\midrule
\multicolumn{4}{@{}l}{mark coherent (never against the measured totals)} & \checkmark \\
\multicolumn{4}{@{}l}{finding: totals tie within spread --- no signal in the mark} & \\
\multicolumn{4}{@{}l}{empty ledger: unmeasured (never zero), no dominance} & \checkmark \\
\multicolumn{4}{@{}l}{one route measured: still no dominance} & \checkmark \\
\bottomrule
\end{tabular}
\end{table}

The books' cost side, calibrated (\Cref{tab:costs}). The instrument is
the CI entrypoint itself --- the capped route-grader --- run five
times over the RISC-V head, each repetition a fresh subprocess with
its own ledger file (a warm content-addressed cache translates
nothing, so a warm run would record nothing: the repetition protocol
is part of the calibration). The pooled profiles put numbers on the
platform's cost shape at probe scale: the two lowering hops and the
bridge each translate in $\sim$0.5\,ms median, the
RISC-V$\to$Sail front is thirty times cheaper, square-oracle checks
run $\sim$1.4\,ms, and the bridged decide dominates everything at
$\sim$3\,ms median --- consistent with the snapshot's finding that
translation is not the cost center (\S\ref{sec:eval-perf}). The
finding is in the middle block: on this host the two routes'
translate totals \emph{tie within the repetition spread} (medians
$\sim$1.1\,ms, $\mu$s apart), and at a tie the report's dominance
mark carries no signal --- it is computed on raw medians with no
noise margin, and across calibration runs it pointed both ways.
What the benchmark certifies instead is the mark's \emph{coherence}
--- pooled and per repetition, a mark never points against the
measured totals --- and the two honesty invariants, run executable:
an empty ledger reads \emph{unmeasured}, never zero, and computes no
dominance; a ledger measuring only one route still computes no
dominance, because dominance needs complete measurement on both
sides. Three honesty notes: the timings are host-tagged and
single-machine by design (profiles never mix machines --- the CI
runner class, where the grader also runs, keeps its own books); the
corpus is the grader's capped probe slice, so these numbers calibrate
the \emph{instrument}, not a workload; and the tie finding names a
concrete next step for the route report --- a declared noise margin
below which no dominance is marked
(\texttt{tools/cost\_calibration.py}; the experiment is also a
permanent test).

\subsection{Post-snapshot: a question campaign for the books}
\label{sec:eval-campaign}

\begin{table}
\caption{The question campaign: 25 authored
questions whose \emph{first failing obstacle is known by
construction} --- the five-obstacle demand taxonomy plus seven
answerable controls spanning every pass mechanism (feasible
routes; an assurance floor met by declared grade; a floor met
by branch corroboration) --- run through the diagnosis against
the live registry with a fresh ledger. ``Targets'' counts the
recommendation board's rows per obstacle: aggregation is per
generation target, so the three cost questions (three sources,
one hub inventory) share one reduction row. Below the rule, the
instrument checks. Post-snapshot measurement at commit
\texttt{e304a40} (2026-07-16).}
\label{tab:campaign}
\footnotesize
\begin{tabular}{@{}lrrr@{}}
\toprule
Obstacle (constructed) & $n$ & diagnosed & targets \\
\midrule
connectivity & 4 & 4/4 & 2 \\
loss & 4 & 4/4 & 4 \\
shape & 4 & 4/4 & 4 \\
cost & 3 & 3/3 & 1 \\
trust & 3 & 3/3 & 3 \\
answerable controls & 7 & 7/7 & --- \\
\midrule
\multicolumn{3}{@{}l}{controls append no demand record} & \checkmark \\
\multicolumn{3}{@{}l}{6 verbatim re-asks add 0 distinct questions} & \checkmark \\
\multicolumn{3}{@{}l}{origins displayed apart (campaign vs.\ organic, 3 rows)} & \checkmark \\
\multicolumn{3}{@{}l}{diagnosis read-only (registry unchanged)} & \checkmark \\
\bottomrule
\end{tabular}
\end{table}

The economy's front half --- diagnosis to recommendation --- measured
(\Cref{tab:campaign}). Twenty-five authored questions are posed
through \texttt{why-not} against the live registry, eighteen of them
constructed so that a \emph{known} obstacle fails first --- routes
that end short of every hub (connectivity), observables no head
projection keeps (loss), question shapes no hub declares (shape),
spent verdicts on feasible routes (cost), and assurance floors that
single-route sources cannot meet (trust) --- beside seven answerable
controls that span every way a question can \emph{pass}: feasible
routes enumerated, a floor met by declared grade
(Python's \texttt{predicted} head), and a floor met by branch
corroboration (both ISAs, and the C spine through its diverse
segment). The diagnosis names the constructed obstacle, first, in all
eighteen cases, with exactly one demand record written per failure;
the seven controls come back answerable and append \emph{nothing} ---
the books never record an answered question. The board then aggregates
the eighteen records to fourteen generation targets, and the
aggregation is per \emph{target}, not per question: the three cost
questions, from three different sources, share the one reduction row
that names the hub's registered abstraction dials. The instrument
checks close the loop: six questions re-asked verbatim add zero
distinct questions (dedup is by question identity); the same question
arriving from an organic session and a synthetic campaign shows both
origins, displayed apart, on its row --- the auditability that keeps
manufactured demand from laundering into evidence (\S\ref{sec:books});
and the diagnosis is read-only against the registry. Two honesty
notes: the corpus is authored, so this measures the diagnosis's
\emph{correctness on constructed ground truth}, not its coverage of
questions players actually ask --- an organic corpus accumulates only
as the platform is played; and the accuracy is of the diagnosis, not
the recommendations --- whether a named target is \emph{worth}
building stays the human judgment the board explicitly declines to
make (\texttt{tools/question\_campaign.py}; the campaign is also a
permanent test).

\subsection{Post-snapshot: the direction axis, measured}
\label{sec:eval-abstraction}

\begin{table}
\caption{The abstraction benchmark. Top, the
authored block: free-set havoc on the \texttt{decoy-$M$}
family (the advisor's zero-loss set) --- the bridge's SMT
artifact (kB) and its decide time (ms), exact$\to$abstracted,
verdicts preserved and the universal transferring. Bottom, the
HWMCC slice: six bit-vector instances from the 2019--2024
corpora, streamed-with-pin (mirror commit
\texttt{57174f5d6f57}, sha256 per instance), each
ingested through the platform's own stack, decided exactly
(btormc, $k{=}20$, \textsc{r}/\textsc{u} =
reachable/unreachable within the bound), then CEGAR-localized:
havoc the advisor's ladder prefix, refine on spurious
counterexamples (each caught by replaying the witness at the
source, havoc inputs filtered), and report where the loop
converges --- states still havocked / total ($n$r = rounds,
$n$s = spurious) --- and that the converged verdict agrees with
the exact one. Post-snapshot measurement at commit
\texttt{b3e1323} (2026-07-16).}
\label{tab:abstraction}
\footnotesize
\setlength{\tabcolsep}{2pt}
\begin{tabular}{@{}lllll@{}}
\toprule
Authored & SMT kB & decide ms & & \\
\midrule
\texttt{decoy-2} & 47.6$\to$6.9 & 68$\to$9 & & \checkmark \\
\texttt{decoy-4} & 92.0$\to$9.6 & 16$\to$7 & & \checkmark \\
\texttt{decoy-8} & 181.9$\to$14.8 & 26$\to$8 & & \checkmark \\
\multicolumn{5}{@{}l}{\texttt{cegar-chain}: 4 spurious, 5 rounds, converges to the free set \checkmark} \\
\multicolumn{5}{@{}l}{\texttt{true-cex}: reachable via the abstraction; source replay confirms \checkmark} \\
\multicolumn{5}{@{}l}{\texttt{sharp-boundary}: a cone state havocked --- spurious; replay refutes \checkmark} \\
\midrule
HWMCC instance & family & exact & localization & \\
\midrule
\texttt{bin-suffix-5} & sosylab'24 & \textsc{u} 0.11s & 0/6 (4r, 3s) & \checkmark \\
\texttt{trex02-1} & sosylab'24 & \textsc{u} 0.11s & 3/7 (1r, 0s) & \checkmark \\
\texttt{benchmark04\_conjunctive} & sosylab'24 & \textsc{u} 0.38s & 0/8 (5r, 4s) & \checkmark \\
\texttt{phases\_2-1} & sosylab'24 & \textsc{r} 0.37s & 0/8 (5r, 4s) & \checkmark \\
\texttt{analog\_estimation} & mann'19 & \textsc{r} 0.01s & 2/5 (1r, 0s) & \checkmark \\
\texttt{adding.5} & beem'19 & \textsc{u} 0.10s & 5/10 (1r, 0s) & \checkmark \\
\bottomrule
\end{tabular}
\end{table}

The direction axis of \S\ref{sec:lax}, exercised end to end
(\Cref{tab:abstraction}). The authored block runs on controlled ground
truth. Havocking the advisor's \emph{free set} --- the states its cone
analysis proves cannot move the question --- preserves the verdict on
every \texttt{decoy-$M$} instance while the bridge's SMT artifact
shrinks by an order of magnitude (the hub encodes every transition it
is handed; localization is what keeps it from being handed dead
update towers --- preparing this measurement caught exactly that: the
translator originally left the orphaned logic in its emission, and the
artifact \emph{grew}; v0.2 sweeps it, a versioned non-extension bump).
Decide time is reported as measured, and its near-parity is a finding,
not a disappointment: engines already refuse to branch into
satisfiable-independent transitions, so the free set is free for
everyone --- what the abstraction buys is the artifact every
downstream encoder, cache, and replay pays for, and the CEGAR
capability inside the cone. The loop itself is the
\texttt{cegar-chain} row: abstracted too aggressively on purpose, it
yields four spurious counterexamples, each caught by source replay
and each un-havocking one ladder rung, converging to exactly the
advisor's free set with the universal verdict transferring
(\Cref{prop:lax}); the \texttt{true-cex} row is the other half of the
asymmetry (a reachable verdict believed only after the source replay
confirms it), and the \texttt{sharp-boundary} control havocs one state
\emph{inside} the cone and must --- and does --- go spurious. The
HWMCC block then takes the same machinery to six bit-vector instances
from the 2019--2024 competition corpora, streamed-with-pin
(BENCHMARKS.md's ingestion discipline: mirror commit plus per-instance
sha256). All six ingest through the platform's own stack --- parse,
shared-evaluator run, reduction advisor --- and preparing the slice
caught a real defect: the beem instance's negated node references
crashed the evaluator with an \emph{untyped} abort and, worse,
silently narrowed the cone analysis (a dependency hidden behind a
negation left the cone --- an unsound free set); both are fixed and
locked by tests, and the instance now decides identically native and
bridged. The localization column is the honest news: on three of six
instances CEGAR converges with states still havocked ---
\texttt{adding.5} stays unreachable with half its protocol state
havocked, \texttt{trex02-1} with 3 of 7 --- while on the other three
it walks all the way back to the exact system, which is what a
competition-curated cone should often force. Two honesty notes: the
free set is empty on every competition instance (curators do not ship
dead state), so all localization here is CEGAR-earned inside the
cone; and six instances at $k{=}20$ measure the machinery, not the
corpus --- HWMCC at competition scale remains future work
(\texttt{tools/abstraction\_bench.py}; the authored block is a
permanent test).

\subsection{Post-snapshot: the player experiment, v2 --- unscripted, over MCP}
\label{sec:eval-player2}

\begin{table}
\caption{The player experiment, v2 --- unscripted,
over MCP: 8 bounded BTOR2 reachability questions
(4/4 polarity split), ground truth established by bridge+btormc
agreement with witnesses replayed, before any run. Arm A reasons
unaided (confidence h/m/l in parentheses); arm B may touch the
platform only through \texttt{gurdy mcp} --- no scripted entry
points, tool discovery included. One fresh agent per question
per arm, no retries. The evidence column is arm B's decisive
artifact. Recorded runs of 2026-07-16; the section formats,
it does not re-run (\texttt{results/llm\_player\_v2/}).}
\label{tab:player2}
\footnotesize
\setlength{\tabcolsep}{2pt}
\begin{tabular}{@{}lllll@{}}
\toprule
Question & Truth & A & B & B's evidence \\
\midrule
Q1: counter, hit@57, $k{=}60$ & \textsc{r} & \textsc{r} (h) & \textsc{r} & replay@57 \\
Q2: same, hit@36, $k{=}30$ & \textsc{u} & \textsc{u} (h) & \textsc{u} & \textsc{unsat} + $k$-bracket \\
Q3: factor semiprime & \textsc{r} & \textsc{r} (h) & \textsc{r} & replay + pinned factors \\
Q4: factor a prime & \textsc{u} & \textsc{r} (l) $\times$ & \textsc{u} & \textsc{unsat} + planted ctl \\
Q5: \texttt{trex02-1} & \textsc{u} & \textsc{u} (h) & \textsc{u} & \textsc{unsat} + advisor + ctl \\
Q6: \texttt{phases\_2-1} & \textsc{r} & \textsc{r} (h) & \textsc{r} & replay@3 \\
Q7: \texttt{adding.5} & \textsc{u} & \textsc{u} (h) & \textsc{u} & \textsc{unsat} + sem.\ probes \\
Q8: LCG, hit@7 & \textsc{r} & \textsc{r} (h) & \textsc{r} & replay@7 + recompute \\
\midrule
\multicolumn{2}{@{}l}{Correct} & \textbf{7/8} & \textbf{8/8} & \\
\bottomrule
\end{tabular}
\end{table}

The first player experiment (\S\ref{sec:eval-player}) named its own
protocol limitations: the entry points were scripted, and twelve
questions at that difficulty could not separate the arms on
correctness. This rerun removes both (\Cref{tab:player2}). The
platform's player surface is now a machine interface --- the MCP
server (\texttt{gurdy mcp}, stdio JSON-RPC; the same enumerated,
advisory tools as the library, exposure pinned by test) --- and arm B
receives only the question, the repository path, and the rule that
this interface is the whole platform: no head-construction snippets,
no named entry points, no reading source or docs. Tool discovery is
part of the task, and every agent performed it cold --- wrote a
stdlib JSON-RPC driver, listed the ten tools, and picked the one
decision call. The questions are harder in exactly one calibrated
way: each is a bounded BTOR2 reachability question given
\emph{verbatim} (constraint-guarded modular counters, 17-bit bounded
factorization, three ingested HWMCC instances including the
negated-refs beem protocol, a 64-bit LCG orbit), with ground truth
established before any run by bridge--btormc agreement and witness
replay. The result is the separation v1 could not produce: \emph{arm
A 7/8, arm B 8/8}, and the miss is exactly where the platform's
floor exceeds unaided cleverness --- asked whether 1676656651
factors in range, the unaided model exhausted trial division and a
Fermat pass, then \emph{guessed} reachable on the meta-heuristic
that benchmark constants are planted semiprimes (the constant is
prime; the run took 34 minutes and was the experiment's only
low-confidence answer --- this time confidence did flag the error,
one instance, not a calibration claim). Arm A's other seven answers
remain formidable --- it decoded two HWMCC pc-encodings, recognized
and bounded the beem adding protocol, and iterated a 64-bit LCG
seven steps by hand, digit-exact --- which sharpens, rather than
dulls, the point: at any difficulty the unaided answer is opaque
recall or heroic arithmetic, graded only by its own confidence,
while every arm-B verdict carries the artifact the calculus grades.
What no protocol scripted is the finding of the experiment: the
players \emph{reinvented the platform's own discipline} --- one
bracketed $k{=}35/36$ to pin the bound convention before trusting a
negative; one pinned its hand-derived factors with added constraints
to machine-check them when the witness exposed no inputs; three
devised positive controls so an \textsc{unsat} could not be a
vacuous pipeline; one probed negated-reference semantics on
self-authored one-bit models before believing the main verdict; and
every negative was hedged, unprompted, as a bounded claim resting on
the bridge --- the exact evidence asymmetry of
\Cref{thm:existential,thm:universal}, articulated by subjects that
were never told it. Honesty notes: the subject model stays in the
builders' family; eight questions measure the protocol, not a
population; ground truth and arm B share the platform's deciders, so
the discriminating arm is A; and one arm-A prompt carried a
transcription slip with its correction inline (recorded), while four
arm-B first launches died on an infrastructure rate limit before any
tool call and were relaunched fresh --- no graded run was retried
(\texttt{results/llm\_player\_v2/}; questions, per-run records, and
all sixteen verbatim final messages).

\subsection{What remains narrated}
\label{sec:eval-exhibits}

Every number above \S\ref{sec:eval-constraint} is the dated snapshot,
regenerable from the artifact by one script; nothing built since is
mixed into those tables. The post-snapshot sections
(\S\ref{sec:eval-constraint}--\S\ref{sec:eval-player2}) are the
former exhibit list, graduated: every family this subsection once
narrated --- constraint enforcement, the books in CI, the demand
board, the refinement demonstration, HWMCC ingestion --- is now a
measured table with its own commit stamp (the trust advisor's
independence judgments are exercised inside the campaign's trust rows
and floor-met-by-corroboration controls,
\S\ref{sec:eval-campaign}), and the player experiment's named
protocol fixes are run (\S\ref{sec:eval-player2}). What remains
future measurement is
\emph{scale}: HWMCC beyond a six-instance slice, an SV-COMP-scale
compliance campaign, organic (rather than authored) question
corpora for the books, and a cross-family subject model for the
player arms.
\fi

\section{Related work}
\label{sec:related}

\paragraph{Certified compilation and translation validation}
CompCert \citep{leroy2009compcert} and CakeML \citep{kumar2014cakeml}
prove one translator correct once; translation validation
\citep{pnueli1998tv,necula2000tv} and its modern descendants
\citep{lopes2021alive2,sewell2013tv} validate each run of one
translator, and validators can themselves be verified
\citep{tristan2008tv}. Credible compilation \citep{rinard1999credible}
--- per-run validation of a compiler's own result --- is a direct
ancestor of re-establishment (\Cref{thm:reestablish}). In our
vocabulary these are grades of a \emph{single edge}; our contribution
is the graph of edges of heterogeneous grade --- how faithfulness,
loss, and assurance compose along routes, how per-run validation
re-establishes hops that carry squares, and how independent routes
corroborate.

\paragraph{Compositional compiler correctness}
Composing simulations across passes and languages is a rich line:
Compositional CompCert \citep{stewart2015compcomp}, parametric
inter-language simulations \citep{neis2015pilsner}, CompCertO's
simulation conventions \citep{koenig2021compcerto}, certified
abstraction layers \citep{gu2015cal}, multi-language semantics
\citep{matthews2007multi,perconti2014open}, and DimSum's decentralized
multi-language graph with wrappers \citep{sammler2023dimsum}. These
works confront vertical composition of simulations \emph{up to
differing observations} --- the habitat of our support condition, which
in their terms is the congruence/well-definedness requirement on the
mediating map; CompCert sidesteps it by fixing one global event type,
CompCertO negotiates it through simulation conventions, DimSum through
wrappers. We claim no mathematical novelty for the condition
(\Cref{def:support} is a one-line congruence); what this paper adds
relative to that line is a different regime and different bookkeeping:
edges that are \emph{not} uniformly proved, trace-equality squares that
are \emph{decidable per run} rather than proved once, an explicit
syntactic keep/loss computation for routes, and the grade/class algebra
for composing heterogeneous evidence --- with the trade-off stated
plainly, since equality-up-to-projection is the right notion only for
near-structure-preserving translations, not optimizing compilation
(\S\ref{sec:languages}).

\paragraph{Semantics-first frameworks}
The K framework derives interpreters, symbolic executors, and provers
from one semantics \citep{rosu2010k}, with KEVM \citep{kevm2018} and
KWasm as flagship instances; Sail does the same for ISAs
\citep{armstrong2019sail}, and ISA-conformance frameworks such as
riscv-formal \citep{wolf2019riscvformal} prefigure the spec-enumerated
per-construct yardsticks of \Cref{def:coverage}. These are our closest
philosophical
relatives, and our platform \emph{consumes} their artifacts (the Sail
RISC-V and Arm models anchor one of each ISA's two routes). The
difference is what carries trust: a semantics-first framework is a
monoculture --- every derived tool inherits the one semantics and the one
derivation pipeline --- whereas our unit of trust is \emph{agreement
between independently derived translations} (\Cref{asm:diversity}),
with the framework structured so routes diverge over independent
semantic artifacts. N-version programming
\citep{avizienis1985nversion} and differential testing
\citep{mckeeman1998difftest,yang2011csmith} exploit the same
independence resource --- and Knight and Leveson's experiment
\citep{knight1986independence} is the standing caution that
independence must be argued, not assumed, which is why
\Cref{asm:diversity} is explicit, scoped to the diverse prefix, and
confronted with our own common-mode incidents (\S\ref{sec:bugs}). We
give the possibilistic account of what agreement does and does not
establish (\Cref{lem:agree}) and wire it into a compositional calculus
rather than a test harness.

\paragraph{Certificates, witnesses, and certifying algorithms}
Proof-carrying code \citep{necula1997pcc}, certifying algorithms
\citep{mcconnell2011certifying}, certified model-checking
\citep{namjoshi2001certifying,yu2021certifaiger,niemetz2018btor2}, and
checked solver proofs (DRAT \citep{wetzler2014drat}, verified checkers
\citep{tan2021cakelpr}, Alethe/Carcara \citep{andreotti2023carcara})
all instantiate the producer/checker split our $\tprov$ grade requires;
SV-COMP's violation/correctness-witness ecosystem
\citep{beyer2015witnesses} is the deployed, cross-tool form of the
existential/universal asymmetry our end-to-end theorems state;
Sledgehammer's checked reconstruction
\citep{blanchette2013sledgehammer} is the same
untrusted-discovery/trusted-replay split inside a proof assistant; and
SMT-COMP's model-validation tracks \citep{weber2019smtcomp}
institutionalize witness re-checking at the solver level. What we
add is the composition through a translation graph: an existential
answer's certificate is its \emph{source-level replay} after
carry-back (\Cref{thm:existential}), while universal answers compose
certificate checking with route fidelity (\Cref{thm:universal}). The
carry-back itself has verified precedent --- decompilation into logic
\citep{myreen2008decompilation} makes the lifter a proved component ---
and our contribution is making it a mandatory, per-pair component of
every edge rather than a per-tool artifact.

\paragraph{Moving programs to answerable representations}
Verified lifting \citep{kamil2016lifting} and the broad practice of
encoding programs into solver logics (BMC, symbolic execution) motivate
the platform's existence; Why3 \citep{filliatre2013why3} is the closest
deployed relative of the hub pattern --- one verification-condition
language translated to many provers of heterogeneous trust --- though
its edges carry no per-run squares and no declared loss; refinement stacks such as seL4's
\citep{klein2009sel4,sewell2013tv} show multi-layer semantic
preservation with uniform, proof-backed layers. Our setting differs in
admitting layers that are \emph{not} uniformly proved --- pinned
compilers, agent-written translators, third-party models --- and making
their heterogeneous trust first-class rather than a threat to the
statement. Combining heterogeneous verification evidence with recorded
provenance is itself prefigured by the Evidential Tool Bus
\citep{rushby2005etb}; the grade/class algebra and the ledger are a
formalized cousin.

\ifarxiv
\paragraph{Abstraction and refinement}
The over-approximating half of the directional square is the
simulation half of abstract interpretation's Galois reading
\citep{cousot1977abstract}; the loop it closes at the platform level
--- ask on the abstraction, transfer the universal verdict, replay the
existential one, refine on a spurious counterexample --- is
counterexample-guided abstraction refinement \citep{clarke2000cegar};
and the first registered instance, havocking caller-named state, is
localization reduction \citep{kurshan1994computer}. What the calculus
adds is not a new abstraction but a different life cycle for one: the
abstraction is a \emph{registered pair} with a declared direction and
witness embedding, its soundness on the corpus is a per-program
decidable check along that embedding rather than a fixed-point
argument, its refinements are durable, graded artifacts rather than
solver-internal state, and its parameters remain the player's ---
advised (cone of influence, observed interval seeds) but never chosen
for them (\S\ref{sec:lax}, \S\ref{sec:books}).
\fi

\paragraph{Fault tolerance, end to end}
The paper's vocabulary is deliberately that of fault-tolerant systems,
because the concepts coincide: the square oracle is an acceptance test
in the recovery-block sense \citep{randell1975recovery}, making a
silently wrong --- in effect Byzantine --- translator fail-stop
\citep{schlichting1983failstop}; a branch is dual modular redundancy
whose comparator localizes rather than votes; and the
existential/universal asymmetry is the end-to-end argument
\citep{saltzer1984endtoend} for translation graphs. What the
transplant adds is the compositional bookkeeping --- graded evidence
with weakest-link meets, per-run re-establishment, and the trusted-base
ledger.

\paragraph{LLMs and formal tools}
Recent systems couple LLM generation with solver-checked validation
\citep{sun2024clover}. Our platform is built for that player --- an LLM
chooses questions and routes but every step it takes is deterministic,
graded, and checked --- and \S\ref{sec:agents} reports the build half
of the two-directional experiment of \Cref{sec:intro}: the translators
themselves were LLM-written, and correctness was carried by the
architecture, not by author trust. \S\ref{sec:eval-player} reports a
first controlled experiment on the player half.

\section{Limitations and conclusion}
\label{sec:conclusion}

\paragraph{Limitations}
This is a dated snapshot of a system designed to ratchet, and its
numbers say so: coverage is deepest on the spine (all 96 RV64IMC
constructs conjoined and composing losslessly to SMT-LIB on both
routes; the in-scope eBPF set complete) and
deliberately partial elsewhere (EVM at 86 of 144 opcodes; Wasm at 54 of
its 75-item declared inventory, without \texttt{loop}/\texttt{br};
Python an integer subset with bounded
unrolling). Coverage is now measured as \Cref{def:coverage}'s actual
conjunction over language-owned inventories wherever a pair has a
decidable square (\S\ref{sec:eval-capability}); the
\texttt{predicted}-grade hops still discharge faithfulness per run
rather than per construct, and the A64 inventory, while language-owned,
is a declared slice of a much larger ISA. Universal answers are bounded claims
wherever BMC-style unrolling caps apply, and rest on
$\mathsf{universal}$-tier hypotheses that per-run evidence and
\S\ref{sec:mech}'s tested surrogates corroborate
but do not entail (\Cref{thm:universal}). The verified portion of the
C spine begins at the compiled binary: the compiler hop has no square
and contributes reproducibility plus a differential, not faithfulness
(\S\ref{sec:composition}). The \tprov\ tier's certified row is inhabited
with a \emph{formally verified} checker at the anchor
(\texttt{cake\_lpr}, \S\ref{sec:eval-proved}), on one trivial and one
certificate-scale instance (an 11.8\,MB LRAT); the residual \tcb\ is
the bit-blaster's BV$\to$CNF step, and a
certified claim is per-question and bounded, never a proof of a pair.
Interpreter adequacy remains an assumption discharged empirically at
modest volume, and the diversity assumption behind branch corroboration
is structural, not probabilistic --- scoped to the diverse prefix, a
boundary the common-mode round (\S\ref{sec:eval-escape}) now measures
directly: shared misreadings blind, by construction, every gate built
from the shared reading; only an independently derived external anchor
catches them.
The compositional core of the calculus --- including the $n$-ary route
telescope and the specialization obligation of \Cref{thm:universal} ---
is mechanized in Lean~4 with a clean axiom audit
(\S\ref{sec:mech})\ifarxiv{} --- the directional square calculus
included, lax telescope and all
(\Cref{prop:exactid,prop:lax})\fi{}; the fieldwise loss computation and the retiming
case remain paper-stated.

\paragraph{Future work}
The gaps above are the roadmap: scale the compliance-derived benchmark
of \S\ref{sec:eval-bench} to an upstream suite taken binary-for-binary
(which needs the machine-mode slice) and to SV-COMP-style C tasks;
extend the escape-rate experiments of \S\ref{sec:eval-escape} beyond
one translator's emission space; widen the A64 slice toward the
base ISA; and scale the player experiment beyond the eight-question
rerun of \S\ref{sec:eval-player2}, with subject models outside the
builders' family. Past these measured gaps
lies the intended trajectory: \hg\ is meant to evolve into an open
platform that scales in language support, where anyone --- working with
LLMs, with agents, or by hand --- develops a new pair against the
calculus's written contract and lands it in the repository through an
ordinary pull request. The admission bar is the architecture, not the
author: a pair arrives with its declared projection\ifarxiv{},
direction,\fi{} and grade, its typed partiality, and a square the
harness runs on merge, and the ratchet (\Cref{prop:ratchet}) keeps
every prior verdict standing as the graph grows.\ifarxiv{} The
direction axis
(\S\ref{sec:lax}) points the growth inward as well as outward: not only
more front-ends, but abstraction and property-transformation pairs ---
endo-edges on the hubs --- whose refinement loops make the graph denser
in \emph{reductions}, with the lax telescope mechanized alongside
the exact one (\S\ref{sec:mech}). Each such refinement is one turn of
the loop of \S\ref{sec:intro} and \S\ref{sec:overview}: diagnose the
missing edge, generate it, gate it, ratchet, re-ask --- growth as a
loop, with the books of \S\ref{sec:books} keeping the demand that
names each next edge, and a human hand on the registration
valve.\fi

\paragraph{Conclusion}
We proposed treating translations the way distributed systems treat
processes: as components that fail, whose failures are contained by
architecture --- declared projections, decidable squares, graded
evidence, independent routes, carried-back witnesses --- rather than
prevented by universal proof. The calculus makes the composition of
partial trust exact; the platform shows the discipline is buildable, by
agents, at measured (if modest) coverage; and the asymmetry theorems say
where proof effort actually buys assurance. The snapshot is the
baseline; the ratchet is the roadmap.

\begin{acks}
This work was co-funded by the Czech Science Foundation under Grant
No.~23-07580X and the European Union under the project Robotics and
Advanced Industrial Production
(reg.~no.~CZ.02.01.01/00/22\_008/0004590).
\end{acks}

\bibliographystyle{ACM-Reference-Format}
\bibliography{references}

\appendix

\section{Proofs}
\label{app:proofs}

\subsection{Pasting (Theorem~3.7)}

We spell out the chain, including the domain side conditions elided in
the paper's proof sketch. Let $p \in \dom(P_2 \circ P_1)$, i.e.\
$p \in \dom(I_A) \cap \dom(T_1)$, $q := T_1(p) \in \dom(T_2)$,
$r := T_2(q)$, $I_C(r)$ defined,
$\carry_2(I_C(r)) \in \dom(\carry_1)$, and (from $p \in \dom(P_1)$,
$q \in \dom(P_2)$) $I_B(q)$ defined with
$I_B(q) \in \dom(\carry_1)$ and $I_C(r) \in \dom(\carry_2)$.

\begin{enumerate}
\item By faithfulness of $P_1$ at $p$ and $\proj \subseteq \proj_1$,
  $\proj(I_A(p)) = \proj(\carry_1(I_B(q)))$.
\item By faithfulness of $P_2$ at $q$:
  $\proj_2(I_B(q)) = \proj_2(\carry_2(I_C(r)))$.
\item Both $I_B(q)$ and $\carry_2(I_C(r))$ lie in $\dom(\carry_1)$ (the
  first by $p \in \dom(P_1)$, the second by
  $p \in \dom(P_2 \circ P_1)$). By
  $(\proj_2 \Rightarrow \proj)$-support applied to step (2):
  $\proj(\carry_1(I_B(q))) = \proj(\carry_1(\carry_2(I_C(r))))$.
\item Chaining (1) and (3):
  $\proj(I_A(p)) = \proj(\carry_1(\carry_2(I_C(r))))$, which is
  faithfulness of $P_2 \circ P_1$ at $p$ w.r.t.\ $\proj$. \qed
\end{enumerate}

\emph{Directional pairs.} For pairs with a non-identity witness
embedding (\Cref{def:pair}), the same four steps run per closing
valuation, the three legs closed at $x$, $W_1(p,x)$, and
$W_2(T_1(p), W_1(p,x))$; the composite's square is then faithfulness
along $W_2 \circ W_1$, with direction the meet $d_1 \wedge d_2$. This
form is machine-checked (\texttt{lax\_pasting}; telescoped,
\texttt{DRoute.\allowbreak lax\_route\_pasting} with
\texttt{DRoute.\allowbreak direction\_exact\_iff}) --- the chain above
is its identity-embedding instance.

\emph{Necessity of support.} Without the support condition
(\Cref{def:support}) the theorem is
false: let $F_B = \{g, h\}$, $\proj_2 = \{g\}$, and let $\carry_1$ copy
field $h$ into a $\proj_1$-kept source field. Take $P_2$ faithful (it
preserves $g$) but with $\carry_2(I_C(r))$ differing from $I_B(q)$ on
$h$ --- permitted, since $h \notin \proj_2$. The outer rectangle then
disagrees on the copied field although both inner squares commute.

\subsection{Composed keep-set, and step granularity}

For fieldwise, step-aligned $\carry_1$ (each kept source field $f$
computed per step as $f = \phi_f(\{\,b.g \mid g \in \mathit{deps}(f)\,\})$),
$(\proj_2 \Rightarrow \proj)$-support holds iff
$\mathit{deps}(f) \subseteq \proj_2$ for every $f \in \proj$; hence the
largest admissible $\proj$ is
$\keep = \{ f \in \proj_1 \mid \mathit{deps}(f) \subseteq \proj_2 \}$ and
support for it is checkable syntactically from the dependency sets.

When $\carry_1$ retimes --- one source step corresponds to a window of
target steps, as in C statements over ISA instructions --- model
$\carry_1$ as a monotone window map $w$ from source step indices to
intervals of target step indices plus a fieldwise read within each
window. Support then requires $\mathit{deps}(f) \subseteq \proj_2$ for
reads \emph{anywhere in the window}, and additionally that the window map
itself be determined by $\proj_2$-observables (e.g.\ by a kept program
counter), since window boundaries influence the projected output. Both
conditions are again syntactic given $w$ and the dependency sets.

\subsection{Localization (Corollary~3.8)}

Contrapositive of the pasting theorem (\Cref{thm:pasting}) for two
hops. For $n$ hops, induct
on the route: if the outer rectangle over hops $1..n$ fails at $p$ but
hop $1$'s square passes at $p$, then by the two-hop case applied to
(hop 1, rectangle over hops $2..n$) the rectangle over $2..n$ fails at
$T_1(p)$; recurse. Termination yields a least failing hop $i$; within
it, the square oracle's comparison yields the least failing step and a
witnessing field. \qed

\subsection{Weakest link (Proposition~3.13)}

Soundness: if every hop is $\mathsf{universal}$ on its fragment, iterate
the pasting theorem (\Cref{thm:pasting}) over all $p$ in the composite
fragment. If some hop is
$\mathsf{perrun}$, the same iteration is available exactly at programs
whose per-hop images passed the oracle --- the per-run set of the route.
If some hop is only $\mathsf{replay}$, purity composes
(\Cref{prop:cache}) but no faithfulness statement is available for the
route beyond it. Optimality: a route with a $\mathsf{perrun}$ hop entails
no universal statement (take a defect outside the checked set); a route
with a $\mathsf{replay}$ hop entails no faithfulness at all (take a pure
but wrong translator). \qed

\subsection{Ratchet (Proposition~5.2)}

Let $I' \sqsupseteq I$ and let $p \in \dom(I)$. Every quantity in
the faithfulness square (\Cref{def:faithful}) evaluated at $p$ reads only
$I$-values on
$\dom(I)$-points, which $I'$ preserves; hence the verdict at $p$ is
unchanged. Coverage monotonicity: $\cov$'s per-construct conjunction can
only gain (new constructs enter $\dom$) and never lose (old verdicts
preserved). Composition of extensions is an extension. For the failure
half: a non-extension update changes some $I(p)$, which can flip the
verdict at $p$ for any pair whose square consumed it --- justifying
the versioned-event rule. \qed

\section{Construct inventories and probe corpora}
\label{app:inventories}

Per-language construct inventories (the yardsticks of the paper's
\Cref{def:coverage}) and the probe corpora behind every coverage figure
in the paper's \Cref{sec:evaluation} are enumerated in the artifact ---
the per-language
inventory modules (\texttt{gurdy/languages/*/inventory.py}), the
per-pair inventories (\texttt{gurdy/pairs/*/inventory.py}), and the
measured coverage data (\texttt{paper/results/data/}); they are
spec-derived: the RISC-V RV64IMC opcode
inventory, the full 144-opcode EVM inventory, the 126-construct eBPF
inventory slice (drawn from the 256-opcode space), the Wasm
numeric/parametric/control inventory, the
OpenSMILES construct list, the Python-subset AST allow-list, and the
two hubs' operator inventories.

\end{document}